\documentclass[letterpaper, 9pt, conference]{ieeeconf}
\IEEEoverridecommandlockouts  
\usepackage{cite}
\usepackage{booktabs}
\usepackage{multirow}
\usepackage[hidelinks]{hyperref}

\usepackage{mathrsfs,amsmath,amssymb,stmaryrd,amsfonts,mathtools}
\usepackage{multirow}
\usepackage[linesnumbered,ruled,vlined]{algorithm2e}
\usepackage{graphicx,algpseudocode}
\usepackage{color}
\usepackage{footnote}
\makesavenoteenv{tabular}
\usepackage{tabularx}
\usepackage{graphicx}
\usepackage{color}
\usepackage{makecell}

\def\defas{\ensuremath{\mathrel{:=}}}

\DeclareMathOperator{\id}{id}

\def\Set#1#2{\ensuremath{
\left\{#1\,\middle|\,#2\right\}
}}

\def\dist{\mathrm{dist}}

\def\dha{\mathbf{d}_{\mathrm{H}}^{\mathrm{a}}}

\def\dh{\mathbf{d}_{\mathrm{H}}^{\mathrm{s}}}

\def\Comp{\mathcal{K}(\mathbb{R}^{n})}

\def\CompDOA{\mathcal{K}(\tilde{\mathcal{D}}_{\mathcal{A}}^{\mathcal{X}})}

\def\norm#1{\|  #1 \| }

\DeclarePairedDelimiter\abs{\lvert}{\rvert}

\DeclareMathDelimiter{\orbrack}{\mathopen}{operators}{"5D}{largesymbols}{"03}
\DeclareMathDelimiter{\clbrack}{\mathclose}{operators}{"5B}{largesymbols}{"02}
\def\intcc#1{\ensuremath{[#1]}}
\def\intoo#1{\ensuremath{\orbrack#1\clbrack}}
\def\intoc#1{\ensuremath{\orbrack#1]}}
\def\intco#1{\ensuremath{[#1\clbrack}}

\def\Hintcc#1{\ensuremath{\llbracket#1\rrbracket}}

\newtheorem{theorem}{Theorem}
\newtheorem{lemma}[theorem]{Lemma}

\newtheorem{remark}{Remark}

\newtheorem{corollary}{Corollary}

\newtheorem{assumption}{Assumption}

\title{Safe and Robust Domains of Attraction for Discrete-Time  Systems: A Set-Based Characterization  and Certifiable Neural Network Estimation}
\author{Mohamed  Serry,  Maxwell Fitzsimmons, and Jun Liu
\thanks{M. S, M. F, and J. L are with the department of Applied Mathematics, University of Waterloo, Waterloo, Ontario, Canada.  
This work was funded in part by the Natural Sciences and Engineering Research Council of Canada  (e-mail: \{mserry,~mfitzsimmons,~j.liu\}@uwaterloo.ca). M.S is also with the Department of Mechanical and Mechatronics Engineering at the University of Waterloo.}}

\date{}

\usepackage{etoolbox}
\makeatletter

\begin{document}
\maketitle

\begin{abstract}
Analyzing nonlinear systems with attracting robust invariant sets (RISs) requires estimating their domains of attraction (DOAs). Despite extensive research, accurately characterizing DOAs for general nonlinear systems remains challenging due to both theoretical and computational limitations, particularly in the presence of uncertainties and state constraints. In this paper, we propose a novel framework for the accurate estimation of safe (state-constrained) and robust DOAs for discrete-time nonlinear uncertain systems with continuous dynamics, open safe sets, compact disturbance sets, and  uniformly locally $\ell_p$-stable compact RISs. The notion of uniform $\ell_p$ stability is quite general and encompasses, as  special cases, uniform exponential and  polynomial stability.  The DOAs are characterized via newly introduced value functions defined on metric spaces of compact sets. We establish their fundamental mathematical properties and derive the associated Bellman-type (Zubov-type) functional equations. Building on this characterization, we develop a physics-informed neural network (NN) framework to learn the corresponding value functions by embedding the derived Bellman-type equations directly into the training process. To obtain certifiable estimates of the safe robust DOAs from the learned neural approximations, we further introduce a verification procedure that leverages existing formal verification tools. The effectiveness and applicability of the proposed methodology are demonstrated through four numerical examples involving nonlinear uncertain systems subject to state constraints, and its performance is compared with existing methods from the literature.

\end{abstract}


\section{Introduction}
The safe robust domain of attraction (DOA) of a discrete-time uncertain dynamical system, with a prescribed attracting  robust invariant set (RIS), is defined as the set of initial states from which all trajectories are guaranteed to converge to that RIS, while satisfying given state constraints for all time steps and for all admissible uncertainties. This set characterizes both the safety and robust asymptotic reachability properties of the system.  DOAs play a critical role in the analysis and design of safe and reliable control strategies, motivating extensive research on their computation and approximation (see, e.g., \cite{liu2025physics,serry2024safe,serry2025safe,chesi2004estimating}). In this paper, we consider the problem of estimating the safe robust DOAs for general discrete-time nonlinear uncertain systems with continuous right-hand sides. Dynamical systems with uncertainties have long posed significant challenges in system analysis and design, as such uncertainties can induce unstable or undesirable behaviors. 

In the context of discrete-time uncertain systems, DOAs are often studied for asymptotically stable equilibrium points (EPs) (i.e., singleton attracting RISs). In the case of unperturbed and unconstrained dynamics, several efforts have been established (see, e.g., \cite{datseris2022effortless}).  Robust DOAs for EPs  are typically estimated using the framework of common Lyapunov functions. Traditional approaches typically assume fixed Lyapunov templates, such as common quadratic Lyapunov functions \cite{amato2002note}, the pointwise maximum of several positive definite quadratic functions \cite{goebel2006dual}, and parameter-dependent Lyapunov functions \cite{coutinho2013local}. The parameters of these functions are usually determined by solving a set of Lyapunov-type matrix inequality conditions (see, e.g., \cite{kau2005new, coutinho2010robust}).  It is worth noting that Lyapunov functions constructed for nominal systems possess a degree of robustness to sufficiently small uncertainties \cite{kalman1960control}, i.e., the asymptotic stability of EPs is preserved under certain classes of small perturbations. More recently, interval-based extensions of Lyapunov conditions have been proposed to estimate DOAs from a given initial Lyapunov function, while explicitly accounting for uncertainty in the system dynamics \cite{lu2024estimating}. The approaches above may yield conservative under-approximations of the true  robust DOAs, largely because they rely on fixed-form Lyapunov function templates. Their performance can degrade even further when state constraints are incorporated, since these fixed templates cannot adapt to or capture the geometry induced by such constraints.

Recently, there has been growing interest in learning-based approaches for estimating DOAs across various classes of systems, where neural networks (typically called neural Lyapunov function) are trained to satisfy Lyapunov-like conditions (see, e.g., \cite{serry2025safe,liu2025physics, chen2021learning}). To certify that the learned neural Lyapunov functions yield valid DOA estimates, formal verification tools, such as those utilizing interval arithmetic and mixed-integer programming, are employed \cite{shi2024certified, wu2023neural}. Although neural networks offer considerable flexibility and training efficiency, their verification remains computationally intensive due to the need for  state space discretization. Nonetheless, recent progress in neural network verification, particularly in linear bound propagation and branch-and-bound techniques, has enhanced the scalability and efficiency of such processes \cite{bunel2020branch, xu2021fast, wang2021beta, ferrari2022complete, zhou2024scalable}.

Safe robust DOAs with prescribed asymptotically stable EPs for discrete-time uncertain systems can be characterized exactly as sublevel sets of suitably defined value functions that satisfy Bellman-type functional equations \cite{xue2020characterization}. 
This characterization offers strong theoretical advantages: if the associated value function is accurately approximated, its sublevel sets can yield tight and potentially large estimates of the safe robust DOA. Existing numerical approaches for solving these functional equations have primarily relied on grid-based discretization methods \cite{xue2020characterization}. 
Such methods are computationally expensive, do not provide formal guarantees due to unquantified discretization errors, and are generally restricted to low-dimensional systems. 
Alternatively, sum-of-squares (SOS) optimization techniques \cite{xue2020robust} have been employed; however, these approaches are limited to systems that admit polynomial representations.  Moreover, the analyses in \cite{xue2020characterization, xue2020robust} rely on several restrictive technical assumptions, including local Lipschitz continuity of the system dynamics and boundedness of the safe set. In addition, pointwise evaluation of the value functions proposed in these works is computationally demanding, as it requires solving infinite-horizon, nonconvex, nonlinear optimization problems. The corresponding Bellman equations also involve technically challenging supremum and infimum operators. Collectively, these difficulties significantly hinder their integration into computationally efficient learning-based frameworks.

In this work, we develop a novel framework for estimating safe robust DOAs associated with prescribed attracting RISs for discrete-time, nonlinear, uncertain systems. 
The proposed approach integrates three key components: 
(i) newly developed theoretical set-based value functions that characterize safe robust DOAs, 
(ii) the expressive approximation capabilities of neural networks for representing these value functions in high-dimensional settings, and 
(iii) recent advances in neural network verification to ensure safety and robustness guarantees of the resulting DOA estimates.

We propose tailored value functions defined on metric spaces of compact sets that characterize the safe robust DOA while permitting tractable evaluation through reachable-set approximations. 
These value functions satisfy associated Bellman-type functional equations, whose residual errors can be practically evaluated via suitable embeddings between set-valued representations and finite-dimensional spaces, thereby enabling physics-informed learning-based computation. 
In contrast to previous approaches, our formulation relaxes several restrictive technical assumptions. 
Specifically, we require only continuity of the system’s right-hand side, openness of the safe set, and uniform $\ell_p$ stability of the RIS, which substantially broadens the scope of applicability.

Although neural network approximations provide flexible and expressive representations of the proposed value functions, they are not sufficient on their own for certifiable DOA estimation, as standard training procedures do not explicitly account for approximation errors. 
To overcome this limitation, we develop a verification framework that certifies DOA estimates obtained from neural network approximations. 
The proposed framework extends the scheme in \cite{serry2025safe} to discrete-time uncertain systems and prescribed RISs, rather than EPs. 
Moreover, it is compatible with state-of-the-art neural network verification tools, including $\alpha,\beta$-CROWN \cite{zhang2018efficient, xu2020automatic, xu2021fast, wang2021beta, zhou2024scalable, shi2024genbab} and dReal \cite{gao2013dreal}.

The organization of this paper is as follows. Section~\ref{sec:Preliminaries} introduces the necessary preliminaries and notation. The problem setup is presented in Section~\ref{sec:ProblemFormulation}. Section~\ref{sec:DOUAProperties} outlines key properties of the safe robust DOAs. The newly proposed value functions and their theoretical properties and characterizations using Bellman-type equations are presented in Section~\ref{sec:ValueFunctions}. Neural network-value function approximations and certifiable DOA estimation are detailed in Section~\ref{sec:VFApproximation}. Section~\ref{sec:NumericalExamples} illustrates the proposed method through numerical examples. Finally, concluding remarks are provided in Section~\ref{sec:Conclusion}.

\section{Notation and preliminaries}
\label{sec:Preliminaries}
Let $\mathbb{R}$, $\mathbb{R}_+$, $\mathbb{Z}$, and $\mathbb{Z}_{+}$ denote the sets of real numbers, nonnegative real numbers, integers, and nonnegative integers, respectively, and define $\mathbb{N} := \mathbb{Z}_{+} \setminus \{0\}$. 
For $a,b \in \mathbb{R}$ with $a \le b$, let $\intcc{a,b}$, $\intoo{a,b}$, $\intco{a,b}$, and $\intoc{a,b}$ denote the closed, open, and half-open intervals in $\mathbb{R}$ with endpoints $a$ and $b$. 
Their discrete counterparts are denoted by $\intcc{a;b}$, $\intoo{a;b}$, $\intco{a;b}$, and $\intoc{a;b}$, respectively, where, for example, $\intcc{a;b} \defas \intcc{a,b} \cap \mathbb{Z}$.  In $\mathbb{R}^{n}
$, the relations $<$, $\leq$, $\geq$, and
$>$ are defined component-wise, e.g., $a < b$, where $a,b\in \mathbb{R}^{n}$, iff $a_i < b_i$ for
all $i\in  \intcc{1;n}$. For $a, b \in \mathbb{R}^n$, with $a \leq b$,
the closed hyper-interval (or hyper-rectangle) $\Hintcc{a,b}$ denotes the set $\Set{x\in \mathbb{R}^{n}}{a\leq x\leq b}$. Let $\norm{\cdot}$ denote any vector norm on $\mathbb{R}^{n}$ and $\mathbb{B}_{n}$ be the $n$-dimensional closed unit ball induced by $\norm{\cdot}$. The $n$-dimensional zero vector is denoted  by $0_{n}$ and the $n$-dimensional vector of ones is denoted by $1_{n}$. Let $\id_{n}$ denote the $n\times n$ identity matrix. For $A\in \mathbb{R}^{n\times m}$, $\norm{A}$  denotes the matrix norm of $A$ induced by the vector norm $\norm{\cdot}$.

The interior and the boundary  of $X\subseteq \mathbb{R}^{n}$ are denoted by $\mathrm{int}(X)$ and $\partial X$, respectively. The Minkowski sum of $X,Y\subseteq \mathbb{R}^{n}$ is defined as
$
X+Y\defas \{x+y,x\in X,~y\in Y\}.
$
 For convenience, we will abuse the above notation for singleton sets, where, for $x\in \mathbb{R}^{n}$ and $Y\subseteq \mathbb{R}^{n}$,  we write 
 $x+Y$ to mean $\{x\}+Y$. 
The class of nonempty compact subsets of $X\subseteq \mathbb{R}^{n}$ is denoted by $\mathcal{K}(X)$.
We have the following standard result:
 \begin{lemma}[See the proof in  Appendix \ref{Proof_lem:enlargement_of_compact}]\label{lem:enlargement_of_compact}
     Let $X\subseteq \mathbb{R}^{n}$ be nonempty and  open and  $\Omega \in \mathcal{K}(X)$. Then there exists $\delta>0$ such that 
     $
\Omega+\delta  \mathbb{B}_{n}\subseteq X.
     $
 \end{lemma}

 Given a point $x\in \mathbb{R}^{n}$ and a set $\Omega\subseteq \mathbb{R}^{n}$, 
 $$
 \dist(x,\Omega)\defas \inf_{y\in \Omega}\norm{x-y}.
 $$
 Let $\dha$ and $\dh$ denote the asymmetric/directed and symmetric Hausdorff distances, respectively, on nonempty bounded subsets of $\mathbb{R}^{n}$, induced by the vector norm $\norm{\cdot}$; That is, 
the asymmetric and symmetric Hausdorff distances between two nonempty
bounded subsets $\Omega, \Gamma \subseteq \mathbb{R}^n$
 are defined as:
$$
\dha(\Omega,\Gamma)\defas\sup_{x\in \Omega}\dist(x,\Gamma),
$$
and
$$
\dh(\Omega,\Gamma)
\defas \max\{\dha(\Omega,\Gamma),\dha(\Gamma,\Omega)\}.
$$
Equivalently, 
$$
\dha(\Omega,\Gamma)=\inf \Set{\varepsilon>0}{\Omega \subseteq \Gamma + \varepsilon \mathbb{B}_{n}},
$$

and
$$
\dh(\Omega,\Gamma)
\defas
\inf
\Set{ \varepsilon > 0}{%
\Omega \subseteq \Gamma + \varepsilon \mathbb{B}_{n},
\Gamma \subseteq \Omega + \varepsilon \mathbb{B}_{n}
}.
$$

When symmetry is not specified, the Hausdorff distance refers to its symmetric version $\dh$. We have the following standard result:

\begin{lemma}[See the proof in  Appendix \ref{Proof_lem:CompactSetsMetric}]\label{lem:CompactSetsMetric}
Let $X\subseteq \mathbb{R}^{n}$. Then $\mathcal{K}(X)$, equipped with the Hausdorff distance  $\dh$, is a metric space. Additionally, if $X$ is open, then $\mathcal{K}(X)$ is open w.r.t. the Hausdorff distance. 
\end{lemma}

Given $f\colon X\rightarrow Y$ and  $P\subseteq X$, the image  of $f$ on $P$ is defined as  $f(P)\defas\Set{f(x)}{x\in P}$.  Given $f\colon X \rightarrow X$ and  $x\in X$,  $f^{0}(x)\defas x$, and for  $M\in \mathbb{N}$, we define $f^{M}(x)$ recursively as follows:   $f^{k}(x)=f(f^{k-1}(x)),~k\in \intcc{1;M}.
$
Given two sets $X$ and $Y$, ${Y}^{X}$ denotes the set of all maps $f\colon X\rightarrow Y$.

Let $f : \mathbb{R}^{n} \to \mathbb{R}^{m}$ be a feedforward neural network with $N$ hidden layers. 
The network is defined recursively as follows. Set
$
z_{0} = x \in \mathbb{R}^{n}.
$
For each hidden layer $k \in \intcc{1;N}$, define
\[
z_{k} = \sigma_k\!\left( \mathtt{W}^{(k)} z_{k-1} + \mathtt{b}^{(k)} \right),
\]
where $\mathtt{W}^{(k)} \in \mathbb{R}^{d_k \times d_{k-1}}$ is the weight matrix, 
$\mathtt{b}^{(k)} \in \mathbb{R}^{d_k}$ is the bias vector, and $d_0 = n$.  The integer $d_k$ denotes the number of neurons in the $k$th hidden layer. For each $k \in \intcc{1;N}$, the activation map
$
\sigma_k : \mathbb{R}^{d_k} \to \mathbb{R}^{d_k}
$
is induced by a scalar function 
$
\varphi_k : \mathbb{R} \to \mathbb{R},
$
and acts componentwise, that is, for every $v = (v_1,\dots,v_{d_k})^\intercal \in \mathbb{R}^{d_k}$,
$
\sigma_k(v)
=
\big(
\varphi_k(v_1), \dots, \varphi_k(v_{d_k})
\big)^\intercal.
$
The output layer is defined by
\[
z_{N+1} = \mathtt{W}^{(N+1)} z_{N} + \mathtt{b}^{(N+1)},
\]
where $\mathtt{W}^{(N+1)} \in \mathbb{R}^{m \times d_N}$ and 
$\mathtt{b}^{(N+1)} \in \mathbb{R}^{m}$. The network mapping is then given by
$
f(x) = z_{N+1}.
$
Equivalently, $f$ can be written as the finite composition of affine maps and componentwise nonlinear activation operators.

\section{Problem formulation}\label{sec:ProblemFormulation}

Consider the discrete-time  uncertain system 
\begin{equation}\label{eq:System}
x_{k+1}=f(x_{k},w_{k}),~k\in \mathbb{Z}_{+},
\end{equation}
where  $x_{k}\in \mathbb{R}^{n}$ is the system state, $w_{k}\in \mathbb{R}^{m}$ is the disturbance or uncertainty, satisfying $
    w_{k}\in W,~k\in \mathbb{Z}_{+},
    $
    where   $W\in \mathcal{K}(\mathbb{R}^{m})$ is a known disturbance set, and  $f\colon \mathbb{R}^{n}\times \mathbb{R}^{m}\rightarrow \mathbb{R}^{n}$ is a continuous function over $\mathbb{R}^{n}\times \mathbb{R}^{m}$ specifying the dynamics of the system. Using the function $f$, we define the map  $F\colon \mathcal{K}(\mathbb{R}^{n})\rightarrow \mathcal{K}(\mathbb{R}^{n})$ as follows:
$$
F(X)\defas f(X,W)=\bigcup_{(x,w)\in X\times W}f(x,w),~X\in \mathcal{K}(\mathbb{R}^{n}).
$$
The map $F$ is well-defined due to  the continuity of $f$ and non-emptiness and  compactness of $W$. In addition: 
\begin{lemma}[See the proof in  Appendix \ref{Proof_lem:FCts}]\label{lem:FCts}
    $F$ is continuous w.r.t. Hausdorff distance $\dh$.
\end{lemma}
A trajectory $\varphi_{x}^{\pi}\colon\mathbb{Z}_{+}\rightarrow \mathbb{R}^{n}$ of system \eqref{eq:System}, induced by an initial condition $x\in \mathbb{R}^{n}$ and a disturbance signal $\pi \colon \mathbb{Z}_{+}\rightarrow W$, satisfies
 $$
 \varphi_{x}^{\pi}(0)=x,~
    \varphi_{x}^{\pi}(k+1)=f(\varphi_{x}^{\pi}(k),\pi(k)),~k\in \mathbb{Z}_{+}.
    $$
The  reachable set of $X\subseteq \mathbb{R}^{n}$ at time $k\in \mathbb{Z}_{+}$, under the dynamics of \eqref{eq:System}, is defined as:
    $$
\mathcal{R}(X,k)\defas\{y\in \mathbb{R}^{n}| \exists (x,\pi) \in X\times W^{\mathbb{Z}_{+}} ~\text{s.t.}~ y=\varphi_{x}^{\pi}(k)\}.
    $$   
The following properties of reachable sets are essential in our analysis in this work:
\begin{lemma}[See the proof in  Appendix \ref{Proof_lem:SemiGroup}]\label{lem:SemiGroup} For all $(X,k)\in \mathcal{K}(\mathbb{R}^{n})\times\mathbb{N},$
    $$
    \mathcal{R}(X,k)=F(\mathcal{R}(X,k-1))=\mathcal{R}(F(X),k-1)=F^{k}(X).
    $$
\end{lemma}

\begin{lemma}[Follows from Lemmas \ref{lem:FCts} and  \ref{lem:SemiGroup}]\label{lem:ReachableSetMap}
    If $X\in \Comp$, then $\mathcal{R}(X,k)\in \Comp$  for all $k\in \mathbb{Z}_{+}$. Also, for fixed $k\in \mathbb{Z}_{+}$, $\mathcal{R}(\cdot,k)\colon \Comp\rightarrow \Comp$ is continuous w.r.t. the Hausdorff distance.
\end{lemma} 
Let $\mathcal{X} \subseteq \mathbb{R}^{n}$ be an open safe set, and let 
$\mathcal{A} \in \mathcal{K}(\mathcal{X})$ be a robust invariant set (RIS), that is,
$$
x\in \mathcal{A}\Rightarrow f(x,W)\subseteq \mathcal{A}.
$$

We impose the following assumption.

\begin{assumption}[Uniform local $\ell_{p}$ stability]\label{Assumptions}
The RIS $\mathcal{A}$ is uniformly locally $\ell_{p}$-stable.
That is, there exist constants
\[
(r,M)\in \intoo{0,\infty}\times \intco{1,\infty},
\]
a number $p\in\intoo{0,\infty}$,
and a function
\[
\lambda:\intcc{0,r}\times\mathbb{Z}_{+}\to \mathbb{R}_{+},
\]
such that:
\begin{enumerate}
\item For each $k\in\mathbb{Z}_{+}$, the mapping 
$s\mapsto \lambda(s,k)$ is continuous and nondecreasing on $\intcc{0,r}$, 
with $\lambda(0,k)=0$. Moreover, for each fixed $s\in\intcc{0,r}$, 
the sequence $k\mapsto \lambda(s,k)$ is nonincreasing on $\mathbb{Z}_{+}$ 
and satisfies $\lambda(s,0)\le s$;
\item The numerical series
\[
\sum_{k=0}^{\infty} \lambda(r,k)^{p}
\]
converges;
\item For all $x\in\mathcal{X}$ with $\dist(x,\mathcal{A})\le r$
and all disturbance sequences $\pi\in W^{\mathbb{Z}_{+}}$,
\[
\dist\!\bigl(\varphi_{x}^{\pi}(k),\mathcal{A}\bigr)
\le
M\,\lambda\!\bigl(\dist(x,\mathcal{A}),k\bigr),
\qquad \forall k\in\mathbb{Z}_{+}.
\]
\end{enumerate}
\end{assumption}

\begin{remark}\label{rem:lp_properties}
Under Assumption~\ref{Assumptions}, the following hold:

\begin{enumerate}
\item For every $s\in\intcc{0,r}$,
$\sum_{k=0}^{\infty} \lambda(s,k)^{p} < \infty$. Indeed, since $\lambda(\cdot,k)$ is nondecreasing,
\(
0\le \lambda(s,k)\le \lambda(r,k)
\)
for all $s\in\intcc{0,r}$, and the comparison test applies.

\item The series
$
\sum_{k=0}^{\infty} \lambda(s,k)^{p}
$
converges uniformly on $\intcc{0,r}$.
This follows from the Weierstrass M-test with
$M_k:=\lambda(r,k)^{p}$.

\item Consequently, the function
$
s\longmapsto \sum_{k=0}^{\infty} \lambda(s,k)^{p}
$
is continuous on $\intcc{0,r}$.

\item For every $s\in\intcc{0,r}$,
$
\lambda(s,k)\to 0$
as $k\to\infty$.

\item Convergence to $0$ and non-negativity imply $\lambda(s,k)\in[0,1]$ for sufficiently large $k$. Hence, if $\bar p\ge p$ then
$
\lambda(s,k)^{\bar p}\le \lambda(s,k)^{p}$,
for sufficiently large $k$, and therefore
$\sum_{k=0}^{\infty} \lambda(s,k)^{\bar p} < \infty
\quad \forall s\in\intcc{0,r}.$

\end{enumerate}
\end{remark}
\begin{remark}[Generality of $\ell_p$ stability]\label{rem:lp_generality}
Assumption~\ref{Assumptions} resembles standard $\mathcal{KL}$-stability assumptions (see, e.g., \cite{kellett2014compendium}), augmented with additional summability requirements. Despite these stronger conditions, $\ell_{p}$-stability still encompasses several common decay regimes, including 
polynomial decay with
\[
\lambda(s,k)
=
\frac{s}{(1+bsk)^{q}},
\]
where $b,q\in\intoo{0,\infty}$,
and exponential decay with
\[
\lambda(s,k)= s\,b^{k},
\]
for some $b\in\intoo{0,1}$.

\end{remark}

The safe robust DOA to $\mathcal{A}$ within $\mathcal{X}$,
$\mathcal{D}_{\mathcal{A}}^{\mathcal{X}}\subseteq \mathbb{R}^{n}$, is defined as
$$
\mathcal{D}_{\mathcal{A}}^{\mathcal{X}}\defas \left\{x\in \mathbb{R}^{n} \middle\vert \begin{array}{c}\forall \pi\in W^{\mathbb{Z}_{+}},~ 
     \varphi_{x}^{\pi}(k)\in \mathcal{X}~ \forall k\in \mathbb{Z}_{+}, \\
      \lim_{k\rightarrow \infty} \dist( \varphi_{x}^{\pi}(k),\mathcal{A})=0\end{array}
\right\}.
$$
Any subset of $\mathcal{D}_{\mathcal{A}}^{\mathcal{X}}$ containing $\mathcal{A}$  and being invariant under the dynamics of \eqref{eq:System}  is called a safe robust region of attraction (ROA) to $\mathcal{A}$ within $\mathcal{X}$. Our goal is to compute a large safe robust ROA that  closely approximates  $\mathcal{D}_{\mathcal{A}}^{\mathcal{X}}$.

\section{Properties of DOA}
\label{sec:DOUAProperties}
In this section, we study the properties of the  safe robust DOA. To do so, we will first study the properties of the safe robust domain of uniform attraction (DOUA):
$$
\tilde{\mathcal{D}}_{\mathcal{A}}^{\mathcal{X}}\defas \left\{x\in \mathbb{R}^{n} \middle\vert \begin{array}{c} \mathcal{R}(\{x\},k)\subseteq  \mathcal{X}~\forall k\in \mathbb{Z}_{+}, \\
       \lim_{k\rightarrow \infty} \dha( \mathcal{R}(\{x\},k),\mathcal{A})=0\end{array}
\right\},
$$
and then we will show at the end of this section that this set is  equal to the safe robust DOA. This equivalence is particularly noteworthy. In continuous-time settings, the DOA is typically shown to coincide only with the interior of the DOUA (see, e.g., \cite{grune2015zubov}), or equality is established under additional structural assumptions, such as convexity of the images of the right-hand side \cite{camilli2001generalization}. In contrast, we establish exact equality in the present discrete-time setting without imposing such restrictive conditions.
\begin{lemma}\label{lem:DOAOpen}$\tilde{\mathcal{D}}_{\mathcal{A}}^{\mathcal{X}}$ is open.
\end{lemma}

\begin{proof}
  Recall the definitions of $M,~r,~\lambda$ in Assumption \ref{Assumptions}. Let $\theta\in \intoo{0,\infty}$ be such that 
    $
\mathcal{A}+\theta \mathbb{B}_{n}\subseteq \mathcal{X},
    $
 due to Lemma \ref{lem:enlargement_of_compact}. Fix $x_{0}\in \tilde{\mathcal{D}}_{\mathcal{A}}^{\mathcal{X}}$. Due to the convergence of $\mathcal{R}(\{x_{0}\},\cdot)$ to $\mathcal{A}$ (w.r.t. $\dha$) within $\mathcal{X}$, there exists $N\in \mathbb{Z}_{+}$ such that  
    $
\mathcal{R}(\{x_{0}\},j)\in \mathcal{X}~\forall j\in \intcc{0;N-1}
    $
and 
    $
\mathcal{R}(\{x_{0}\},N)\subseteq\mathcal{A}+({\tilde{r}}/{2})\mathbb{B},
    $
where  $\tilde{r}\in \mathbb{R}_{+}$ satisfies
$$
0<\tilde{r}\leq \min \left \{\frac{\theta}{M}, r\right \}.
$$
Let $\rho_{j}\in \mathbb{R}_{+},~j\in \intcc{0;N-1}$, be positive numbers satisfying
$
\mathcal{R}(\{x_{0}\},j)+\rho_{j}\mathbb{B}_{n}\subseteq \mathcal{X},~j\in \intcc{0;N-1}.
$
Such numbers exist due to the nonemptiness and  openness of $\mathcal{X}$, the  compactness of the values of $\mathcal{R}(\{x_{0}\},\cdot)$, and Lemma \ref{lem:enlargement_of_compact}.
Recall from Lemma \ref{lem:SemiGroup} that 
$
\mathcal{R}(\{y\},j)=F^{j}(\{y\})
$
for all $y\in \mathbb{R}^{n}$ and $j\in \mathbb{Z}_{+}$. Using  the continuity result from Lemma \ref{lem:ReachableSetMap}, there exists $\delta\in \intoo{0,\infty}$ such that, for all $x\in x_{0}+\delta \mathbb{B}_{n}$, 
$$
\dh(F^{j}(\{x\}),F^{j}(\{x_{0}\})) \leq  \rho_{j},~j\in \intcc{0;N-1},
$$
and 
$$
\dh(F^{N}(\{x\}),F^{N}(\{x_{0}\}))\leq \frac{\tilde{r}}{2}.
 $$
Consequently, we have, for all $x\in x_{0}+\delta \mathbb{B}_{n}$,
\begin{align*}
\mathcal{R}(\{x\},j)=F^{j}(\{x\})\subseteq \mathcal{R}(\{x_{0}\},j)+\rho_{j}\mathbb{B}_{n}\subseteq \mathcal{X},
\end{align*}
$j\in \intcc{0;N-1}$, and
\begin{align*}
\mathcal{R}(\{x\},N)=&F^{N}(\{x\})
\subseteq F^{N}(\{x_{0}\})+({\tilde{r}}/{2})\mathbb{B}_{n}\\
& \subseteq\mathcal{A}+\frac{\tilde{r}}{2}\mathbb{B}_{n}+\frac{\tilde{r}}{2}\mathbb{B}_{n}=\mathcal{A}+\tilde{r}\mathbb{B}_{n}.
\end{align*}
The local $\ell_{p}$ stability then indicates that, for all $x\in x_{0}+\delta \mathbb{B}_{n}$, 
    \begin{align*}
\dha(F^{N+k}(\{x\}),\mathcal{A})&\leq M \lambda( \dha(F^{N}(\{x\}),\mathcal{A}),k) \leq M\lambda(\tilde{r},k),
    \end{align*}
$k\in \mathbb{Z}_{+}$. Hence, for all $x\in x_{0}+\delta \mathbb{B}_{n}$,
$\lim_{j\rightarrow \infty }\dha(\mathcal{R}(\{x\},j),\mathcal{A})=0.
$
The above inequality and    the definition of $\tilde{r}$ imply that, for all $x\in x_{0}+\delta \mathbb{B}_{n}$,
$$
F^{N+k}(\{x\})\subseteq  \mathcal{A}+ M \tilde{r} \mathbb{B}_{n}\subseteq   \mathcal{A}+\theta \mathbb{B}_{n}\subseteq \mathcal{X},~k\in \mathbb{Z}_{+}.
$$
    Therefore, $x\in\tilde{\mathcal{D}}_{\mathcal{A}}^{\mathcal{X}}$ for all $x\in x_{0}+\delta \mathbb{B}_{n}$. As $x_{0}\in \tilde{\mathcal{D}}_{\mathcal{A}}^{\mathcal{X}}$ is arbitrary, the proof is complete.
\end{proof}
The proof of the lemma above can be easily adapted to show the following:
\begin{lemma} \label{lem:delta-margin}
    For each $x\in \tilde{\mathcal{D}}_{\mathcal{A}}^{\mathcal{X}}$, there exists $\delta_{x}>0$ (that depends on $x$) such that 
    $
\mathcal{R}(x+\delta_{x} \mathbb{B}_{n},k)\subseteq \mathcal{X}~\forall k\in \mathbb{Z}_{+}
    $
    and 
    $
\lim_{k\rightarrow \infty}\dha(\mathcal{R}(x+\delta_{x} \mathbb{B}_{n},k),\mathcal{A})=0.
    $
\end{lemma}
Building upon the above result, we have the following important  property:
\begin{lemma}\label{lem:CompactSubsetsOfDOUA}
    For any $X\in \CompDOA$, we have  
        $
\mathcal{R}(X,k)\subseteq \mathcal{X}~\forall k\in \mathbb{Z}_{+},
    $
    and 
    $
\lim_{k\rightarrow \infty}\dha(\mathcal{R}(X,k),\mathcal{A})=0.
    $
\end{lemma}
\begin{proof}
For each $x\in X$, let $\delta_{x}>0$ be defined as in Lemma \ref{lem:delta-margin}. We have 
$
X \subseteq \bigcup_{x\in X}x+\delta_{x}\mathbb{B}_{n}. 
$
 The compactness of $X$ implies the existence of finite $N\in \mathbb{N}$ and a finite set $\{x_{i}\}_{i=1}^{N}\subseteq X$ such that 
    $
X\subseteq \bigcup_{i\in \intcc{1;N}}x_{i}+\delta_{x_{i}}\mathbb{B}_{n}.
    $
    We have, using Lemma \ref{lem:delta-margin}, 
    $$
\mathcal{R}(\bigcup_{i\in \intcc{1;N}}x_{i}+\delta_{x_{i}}\mathbb{B}_{n},k)=\bigcup_{i\in \intcc{1;N}}\mathcal{R}(x_{i}+\delta_{x_{i}}\mathbb{B}_{n},k)\subseteq \mathcal{X},
    $$
    $ k\in \mathbb{Z}_{+}$, and 
    \begin{align*}
\lim_{k\rightarrow \infty}\dha(\mathcal{R}(\bigcup_{i\in \intcc{1;N}}x_{i}+\delta_{x_{i}}\mathbb{B}_{n},k),\mathcal{A})\leq\\ \lim_{k\rightarrow \infty} \sup_{i\in \intcc{1;n}}\dha(x_{i}+\delta_{x_{i}}\mathbb{B}_{n},\mathcal{A})=0.
    \end{align*}
\end{proof}
Next, we show the invariance of the safe robust DOUA under the dynamics of \eqref{eq:System}:
\begin{lemma}\label{lem:compDOA_inv}
    For all $X \in \CompDOA$, $F(X)\in \CompDOA$  (invariance of $\CompDOA$ under $F$). 
\end{lemma}

\begin{proof}
Using Lemma \ref{lem:CompactSubsetsOfDOUA}, we have,   for   $X\in \CompDOA$, $\mathcal{R}(X,k)\in \mathcal{X}~\forall k\in \mathbb{Z}_{+},$ and $\lim_{k\rightarrow \infty}\dha(\mathcal{R}(X,k),\mathcal{A})=0.$
    This implies, using the semi-group property (Lemma \ref{lem:SemiGroup}), that  
    $
    \mathcal{R}(X,k+1)=\mathcal{R}(F(X), k)\in \mathcal{X}~\forall k\in \mathbb{Z}_{+},
    $ and 
    $
\lim_{k\rightarrow \infty}\dha(\mathcal{R}(X,k+1),\mathcal{A})=\lim_{k\rightarrow \infty}\dha(\mathcal{R}(F(X),k),\mathcal{A})=0.
  $
  Hence, $F(X)\in \tilde{\mathcal{D}}_{\mathcal{A}}^{\mathcal{X}}$.
\end{proof}   

Now, we show the equivalence between the DOA and the DOUA.

\begin{theorem}
    $
\mathcal{D}_{\mathcal{A}}^{\mathcal{X}}=\tilde{\mathcal{D}}_{\mathcal{A}}^{\mathcal{X}}.
    $
\end{theorem}
  \begin{proof}
    Define $\hat F:\mathcal X\to \mathcal K(\mathcal X)$ to be $\hat F(x)=F(\{x\}),~x\in \mathcal{X}$. 
    Then, $\hat F$ is a compact valued continuous multifunction. 
    Moreover, define $\hat F^{-}(V)\defas\{x\in \mathcal X: \hat F(x)\cap V\neq \emptyset \}$ and $\hat F^{+}(V)\defas\{x\in \mathcal X: \hat F(x)\subseteq V\}$ for $V\subseteq \mathcal{X}$. 

    First, note that $\mathcal{D}_{\mathcal{A}}^{\mathcal{X}}$ is invariant under $\hat F$, that is, for all $x\in \mathcal{D}_{\mathcal{A}}^{\mathcal{X}}$ we have $\hat F(x)\subseteq \mathcal{D}_{\mathcal{A}}^{\mathcal{X}}$. 
    Second, note that $\mathcal D^{\mathcal X}_{\mathcal{A}}\supseteq  {\mathcal {\tilde D}^{\mathcal X}_{\mathcal{A}}}$.
    These facts follow quickly from the definitions and the preceding lemmas. 
    So we only need to show that $\mathcal D^{\mathcal X}_{\mathcal{A}}\subseteq  {\mathcal {\tilde D}^{\mathcal X}_{\mathcal{A}}}$. 
    
    Let $\Omega={\mathcal { D}^{\mathcal X}_{\mathcal{A}}}\setminus {\mathcal {\tilde D}^{\mathcal X}_{\mathcal{A}}}$. We claim that $\Omega\subseteq \hat F^{-}(\Omega)$. 
    To see this, if $x\in \Omega\cap (\mathbb R^n\setminus \hat F^{-}(\Omega))$, we get $\hat F(x)\subseteq \mathcal {D}^{\mathcal X}_{\mathcal{A}}$ (as $x\in \mathcal { D}^{\mathcal X}_{\mathcal{A}}$) and $\hat F(x) \subseteq \mathbb R^n\setminus \Omega$ (this immediately follows from the definition of $\hat F^{-}$). 
    Thus, $\hat F(x) \subseteq \mathcal { D}^{\mathcal X}_{\mathcal{A}}\cap  \mathbb R^n\setminus \Omega =\mathcal { D}^{\mathcal X}_{\mathcal{A}}\setminus (\mathcal { D}^{\mathcal X}_{\mathcal{A}}\setminus{\mathcal {\tilde D}^{\mathcal X}_{\mathcal{A}}})=\mathcal {\tilde D}^{\mathcal X}_{\mathcal{A}}$.
    Since $\hat F(x)$ is compact, we have $\hat F(x)\in \mathcal K(\mathcal {\tilde D}^{\mathcal{X}}_{\mathcal{A}})$. 
    Applying Lemma~\ref{lem:CompactSubsetsOfDOUA} with $X$ replaced by $F(\{x\})=\hat F(x)$ and recalling Lemma~\ref{lem:SemiGroup}, we find that $x\in \mathcal {\tilde D}^{\mathcal{X}}_{\mathcal{A}}$. 
    But this contradicts the definition of $\Omega$. 
    Therefore, $\Omega\subseteq \hat F^{-}(\Omega)$.

    Any set $B\subseteq \mathbb R^n$  satisfying $B\subseteq \hat F^{-}(B)$ is often called a viable set (of $\hat F$). Viable sets are often considered dual to invariant sets. 
    More importantly, $B$ is viable if and only if for all $x\in B$ there is $\pi \in W^{\mathbb Z_+}$ with $\varphi_x^\pi(k) \in B~\forall k\in \mathbb{Z}_{+}$. 
    See~\cite[Proposition~3.1.8]{Fitzsimmons2023Properties} for the proof. 
    
    The openness of $\mathcal {\tilde D}^{\mathcal X}_{\mathcal{A}}$ and the compactness of $\mathcal{A}$ imply the existence of some $\delta>0$ such that $\mathcal{A}+\delta \mathbb{B}_{n}\subset \mathcal {\tilde D}^{\mathcal X}_{\mathcal{A}}$ (Lemma \ref{lem:enlargement_of_compact}). Therefore, it follows that $\Omega\subseteq \mathcal X \setminus (\mathcal{A}+\delta \mathbb B_n$). 
    Thus any trajectory which remains in $\Omega$ cannot converge to $\mathcal{A}$ w.r.t. $\dist$. 
    Should $\Omega\neq \emptyset$ this would be a contradiction (As $\Omega\subseteq \mathcal {D}^{\mathcal X}_\mathcal{A}$, so every trajectory starting in $\Omega$ converges to $\mathcal{A}$). 
    Thus $\Omega=\emptyset$ and the conclusion follows.
\end{proof}

\section{Value functions} 
\label{sec:ValueFunctions}
In this section, we introduce novel value functions, defined on metric spaces of compact sets, which characterize the safe robust DOA (equivalently, the safe robust DOUA). In our construction of these value functions, we define the  functions $$\alpha\colon \mathbb{R}^{n}\rightarrow \mathbb{R}_{+}$$ and $$\gamma \colon \mathbb{R}^{n} \rightarrow \mathbb{R}_{+}\bigcup\{\infty\},
$$
satisfying the following properties: 
\begin{itemize}
    \item 
    $\alpha$ is continuous, satisfying
\begin{equation}\label{eq:AlphaBounds}
\underline{\alpha}\dist(x,\mathcal{A})^{\bar{p}}\leq\alpha(x)\leq \overline{\alpha}\dist(x,\mathcal{A})^{\bar{p}},~x\in \mathbb{R}^{n},
\end{equation}
for some $\underline{\alpha}, \overline{\alpha}\in \intoo{0,\infty}$ and $\bar{p}\in \intco{p,\infty}$. 

\item $\gamma$ is finite and  continuous over $\mathcal{X}$ such that there exists  $\underline{\gamma}\in \intoo{0,\infty}$ with  
        \begin{equation}\label{eq:GammaBound}
\gamma(x)\geq \underline{\gamma}~\forall x\in \mathcal{X}, 
        \end{equation}
         $$x\notin \mathcal{X}\Rightarrow \gamma(x)=\infty,
         $$
        and, for any sequence $\{x_{k}\}\subseteq \mathbb{R}^{n}$,  $$x_{k}\rightarrow x\in \partial \mathcal{X}\Rightarrow \gamma(x_{n})\rightarrow \infty.
        $$
   
  \end{itemize}   
\begin{remark}\label{rem:Alpha}
The existence  of $\alpha$ with the aforementioned properties is guaranteed (for example,  we can set $\alpha(x)= \dist(x,\mathcal{A})^{p},~x\in \mathbb{R}^{n}$).
\end{remark}
    \begin{remark}\label{rem:Gamma} \label{rem:transformation}
    The existence  of the  function $\gamma$ is  guaranteed for  various classes of sets.  For instance, if $\mathcal{X}$ is a strict 1-sublevel set of a continuous function $g\colon \mathbb{R}^{n}\rightarrow \mathbb{R}$, i.e., $\mathcal{X}=\{x\in \mathbb{R}^{n}|g(x)<1\},$
    then we can define 
    $$
\gamma(x)\defas \gamma_{0}+\frac{1}{\mathrm{ReLu}(1-g(x))},
$$
    where $1/0\defas \infty$, $\mathrm{ReLU}\colon \mathbb{R}\rightarrow \mathbb{R}$ is the rectifier linear unit function defined as $\mathrm{ReLU}(x)=(x+|x|)/2,~x\in \mathbb{R}$, and ${\gamma_{0}}>0$ is any arbitrary positive constant. With this definition of $\gamma$, the lower bound condition holds with $\underline{\gamma}=\gamma_{0}$. For $\mathcal{X}=\mathbb{R}^{n}$, we can simply set $\gamma(\cdot)=1$. 
\end{remark}
\begin{remark}\label{rem:Superposition}
Suppose the safe set $\mathcal{X}$ is given as the intersection of 
$N_{\mathrm{safe}}$ open sets $\mathcal{X}_{i}$, 
$i \in \intcc{1;N_{\mathrm{safe}}}$, i.e.,
\[
\mathcal{X}
=
\bigcap_{i=1}^{N_{\mathrm{safe}}}
\mathcal{X}_{i},
\]
where each $\mathcal{X}_{i}$ contains the RIS $\mathcal{A}$ and is defined as a strict $1$-sublevel set of a continuous function 
$g_{{i}} : \mathbb{R}^n \to \mathbb{R}$, namely,
$\mathcal{X}_{i}
=
\Set{x \in \mathbb{R}^n}{g_{{i}}(x) < 1}$.
Then $\mathcal{X}$ itself can be represented as a strict $1$-sublevel set of the function 
$g : \mathbb{R}^n \to \mathbb{R}$ defined by
$
g(x)
\defas
\max_{i \in \intcc{1;N_{\mathrm{safe}}}}
g_{{i}}(x),~ x \in \mathbb{R}^n$.

Moreover, the function $\gamma$ may be constructed using the obtained function $g$ as described in Remark~\ref{rem:Gamma}. 
Alternatively, a superposition-based construction can be employed: 
for each $\mathcal{X}_{i}$, compute $\gamma_{i}$ according to Remark~\ref{rem:Gamma}, and define
$\gamma(\cdot)
\defas
\sum_{i=1}^{N_{\mathrm{safe}}}
\gamma_{{i}}(\cdot)$.
\end{remark}

Define the value functions $\mathcal{V},\mathcal{W} \colon \mathcal{K}(\mathbb{R}^{n})\rightarrow \mathbb{R}_{+}\bigcup \{\infty\}$ as follows:
\begin{align}\label{eq:Lyapunov1}
\mathcal{V}(X)&\defas \sum_{k=0}^{\infty} \Psi(\mathcal{R}(X,k)),\\
\label{eq:Lyapunov2}
\mathcal{W}(X)&\defas 1-\exp(-\mathcal{V}(X)),
\end{align}
where $\Psi\colon \Comp\rightarrow \mathbb{R}_{+}\bigcup\{\infty\}$ is given as:
\begin{equation}\label{eq:Psi}
\Psi(X)\defas \sup_{y_{1}\in X}\gamma(y_{1})\sup_{y_{2}\in X}\alpha(y_{2}),~X\in \Comp, 
\end{equation}
and we define $\exp(-\infty)\defas 0$. The presented value functions characterize the safe robust DOUA (equivalently, DOA) as shown below:
\begin{theorem}\label{thm:DOASublevel}
    $
\mathbb{V}_{\infty}\defas \{x\in \mathbb{R}^{n}| \mathcal{V}(\{x\})<\infty\}=\mathbb{W}_{1}\defas\{x\in \mathbb{R}^{n}| \mathcal{W}(\{x\})<1\} =\tilde{\mathcal{D}}_{\mathcal{A}}^{\mathcal{X}}.
    $
\end{theorem}

\begin{proof}
It is sufficient to show that $\mathbb{V}_{\infty}=\tilde{\mathcal{D}}_{\mathcal{A}}^{\mathcal{X}}$, as $\mathcal{W}(\{x\})<1\Leftrightarrow\mathcal{V}(\{x\})<\infty,~x\in \mathbb{R}^{n}$.     Let $x\in \mathbb{V}_{\infty}$. We will show that $\mathcal{R}(\{x\},k)\subset \mathcal{X}$ for all $k\in \mathbb{Z}_{+}$.  By contradiction, assume that $\mathcal{R}(\{x\},k)\not\subseteq \mathcal{X}$ for some $N\in \mathbb{Z}_{+}$. Then, by the definition of $\gamma$, 
    $
\sup_{y\in \mathcal{R}(\{x\},k)}\gamma(y)=\infty.
    $
    We note that $\sup_{y\in \mathcal{R}(\{x\},k)}\alpha(y)>0$ if $\mathcal{R}(\{x\},k)\not\subseteq \mathcal{A}$ by  \eqref{eq:AlphaBounds}. Hence, $\Psi(\mathcal{R}(\{x\},k))=\infty$, implying  $\mathcal{V}(\{x\})=\infty$, and that contradicts the fact that $x\in \mathbb{V}_{\infty}$. Hence, we have $\mathcal{R}(\{x\},k)\subseteq \mathcal{X}$ for all $k\in \mathbb{Z}_{+}$. Using \eqref{eq:GammaBound}, we have 
$$
\sup_{y\in \mathcal{R}(\{x\},k)}\alpha(y)\leq \frac{\Psi(\mathcal{R}(\{x\},k))}{\underline{\gamma}},~k\in \mathbb{Z}_{+}.
$$
As the infinite sum $
\sum_{k=0}^{\infty}\Psi(\mathcal{R}(\{x\},k))
$ is convergent, 
$\sum_{k=0}^{\infty}\sup_{y\in \mathcal{R}(\{x\},k)}\alpha(y)$
is also convergent by the  comparison test, implying  $\lim_{k\rightarrow \infty} \sup_{y\in \mathcal{R}(\{x\},k)}\alpha(y)=0$. Using the lower bound on $\alpha$, we have
    \begin{align*}
\lim_{k\rightarrow \infty} \dha(\mathcal{R}(\{x\},k),\mathcal{A})^{\bar{p}}\leq \lim_{k\rightarrow \infty} \frac{1}{\underline{\alpha}}\sup_{y\in \mathcal{R}(\{x\},k)}\alpha(y)=0.
    \end{align*}
    Therefore, $\lim_{k\rightarrow \infty} \dha(\mathcal{R}(\{x\},k),\mathcal{A})=0$, and consequently, $x\in\tilde{\mathcal{D}}_{\mathcal{A}}^{\mathcal{X}}$.

Now, let $x\in\tilde{\mathcal{D}}_{\mathcal{A}}^{\mathcal{X}}$ and  recall the definitions of $M,~r,~\lambda$ in Assumption \ref{Assumptions}.  Using the properties of $\alpha$ and $\gamma$, we see that $\sup_{y\in \mathcal{R}(\{x\},j)}\alpha(y)<\infty$ and  $
0<\underline{\gamma}\leq \sup_{y\in \mathcal{R}(\{x\},j)}\gamma(y)<\infty$ for all $j\in \mathbb{Z}_{+}$ as $\mathcal{R}(\{x\},j)\in \mathcal{K}(\mathcal{X}),~j\in \mathbb{Z}_{+}$ (the supremum of a continuous function over a compact subset of the domain is finite). Consequently, we have $\Psi(\mathcal{R}(\{x\},j))<\infty$ for all $j\in \mathbb{Z}_{+}$.   Let $\theta\in \intoo{0,\infty}$ be such that 
    $ \mathcal{A}+
\theta \mathbb{B}_{n}\subseteq \mathcal{X},
    $
    which exists due to the openness of  $\mathcal{X}$, the compactness of $\mathcal{A}$, and Lemma \ref{lem:enlargement_of_compact}. Since $\mathcal{R}(\{x\},\cdot) $ converges to $\mathcal{A}$ w.r.t. $\dha$, there exists $N\in \mathbb{Z}_{+}$ such that 
$
\mathcal{R}(\{x\},N)\subseteq  \mathcal{A}+\tilde{r}\mathbb{B}_{n},
$
where 
$
0<\tilde{r}\leq \min \{{\theta}/{M},r\}.
$
As $\mathcal{R}(\{x\},N)\subseteq \mathcal{A}+  {r}\mathbb{B}_{n}$ ($\tilde{r}\leq r$), the local $\ell_{p}$ stability and the definition of $\tilde{r}$ imply that 
\begin{align*}
\dha(\mathcal{R}(\{x\},N+k), \mathcal{A})&\leq M \lambda(\dha(\mathcal{R}(\{x\},N),\mathcal{A}),k)\\
&\leq M \lambda(\tilde{r},k)
\leq  M \frac{\theta}{M}=\theta  ,~k\in \mathbb{Z}_{+},
\end{align*}
and
\begin{align*}
\sup_{y\in \mathcal{R}(\{x\},N+k)}\alpha(y)&\leq \overline{\alpha}\dha(\mathcal{R}(\{x\},N+k),\mathcal{A})^{\bar{p}}\\
&\leq \overline{\alpha}M^{\bar{p}}(\lambda(\tilde{r},k))^{\bar{p}},~k\in \mathbb{Z}_{+}.
\end{align*}
Let $\Gamma_{\theta }\in \intoo{0,\infty}$ be such that 
$
\gamma(x)\leq \Gamma_{\theta }~\forall x\in \theta \mathbb{B}_{n},
$
which exists due to the compactness of $\mathcal{A}+\theta \mathbb{B}_{n}$ and the continuity of $\gamma$ over the compact $\mathcal{A}+\theta \mathbb{B}_{n}$. Therefore,  using the summability property of $\lambda$, we have  
\begin{align*}
\mathcal{V}(\{x\})=&\sum_{k=0}^{\infty}\Psi(\mathcal{R}(\{x\},k))\\
=&\sum_{k=0}^{N-1}\Psi(\mathcal{R}(\{x\},k))+\sum_{k=N}^{\infty}\Psi(\mathcal{R}(\{x\},k))\\
&\leq \sum_{k=0}^{N-1}\Psi(\mathcal{R}(\{x\},k))+\overline{\alpha}M^{\bar{p}}\Gamma_{\theta}\sum_{k=0}^{\infty}(\lambda(\tilde{r},k))^{\bar{p}}\\&<\infty.
\end{align*}
Hence, $x\in \mathbb{V}_{\infty}$, which completes the proof.
\end{proof}
Using Lemma \ref{lem:CompactSubsetsOfDOUA} and by adapting the proof of the theorem above, we have 
\begin{lemma}
    For a compact $X$, $X\in \CompDOA$ if and only if $\mathcal{V}(X)<\infty$.
\end{lemma}



\subsection{Properties of the value functions}
\label{sec:VFProperties}
In this section, we state some important properties for the functions $\mathcal{V}$ and $\mathcal{W}$.

\begin{lemma}
    The functions $\mathcal{V}$ and $\mathcal{W}$ are positive definite in the sense that they are strictly positive for any nonempty compact $X\not \subseteq \mathcal{A}$, and are identically zero when $X\subseteq \mathcal{A}$.
\end{lemma}
\begin{proof}
This is an immediate consequence of the definitions.
\end{proof}
\begin{lemma}\label{lem:Monotonicity}
    The functions $\mathcal{V}$ and $\mathcal{W}$ are monotonic w.r.t. set inclusion. That is, for $X,Y\in \Comp$, 
    $
    X\subseteq Y\Rightarrow \mathcal{V}(X)\leq \mathcal{V}(Y)$ and $\mathcal{W}(X)\leq \mathcal{W}(Y).
    $
\end{lemma}
\begin{proof}
    This is an immediate consequence of the definitions.
\end{proof}
Next, it is our goal to prove the continuity of $\mathcal{V}$. Below is a technical result that will be used in our continuity analysis.
\begin{lemma}[See the proof in  Appendix \ref{Proof_Lem:SupIsCts}]\label{Lem:SupIsCts}
Let $g\colon \mathbb{R}^{n}\rightarrow \mathbb{R}$ be continuous and define $G\colon D \subseteq \mathcal{K}(\mathbb{R}^{n})\rightarrow \mathbb{R}$  as 
$
G(X)\defas \sup_{x\in X}g(x),~X\in D.
$
Then $G$ is continuous over $D$.
\end{lemma}

\begin{theorem}\label{thm:VCts}
    $\mathcal{V}$ is continuous and finite over $\CompDOA$.
\end{theorem}

\begin{proof}
     Recall the definitions of $M,~r,~\lambda$ in Assumption \ref{Assumptions} and the definition of $\Psi$ in \eqref{eq:Psi}. Let $\theta\in \intoo{0,\infty}$ be such that 
    $
\mathcal{A}+\theta \mathbb{B}_{n}\subseteq \mathcal{X}
    $ (existence of such $\theta$ has been established in previous proofs). Let $X\in \CompDOA$ and $\varepsilon>0$ be arbitrary, where we assume without loss of generality that 
    $\varepsilon\leq \min\{{\theta}/{M},r\}$. 
     Then, using Lemma \ref{lem:CompactSubsetsOfDOUA},  there exists $N\in \mathbb{Z}_{+}$ such that
    $
\mathcal{R}(X,k)\subseteq \mathcal{X}~\forall k\in \mathbb{Z}_{+},
    $
    with
    $\mathcal{R}(X,N)\subseteq \mathcal{A}+ ({\varepsilon}/{2})\mathbb{B}_{n}$.
    The properties of $\alpha$ and $\gamma$ and the containment of the reachable sets, which are compact, inside the safe set imply that
$
\sum_{k=0}^{N-1}\Psi(\mathcal{R}(X,k))<\infty.
$
     The $\ell_{p}$ stability indicates that 
     $$
\dha(\mathcal{R}(X,N+k),\mathcal{A})\leq M \lambda(\varepsilon,k)  \leq M \frac{\theta}{M}=\theta ,~k\in \mathbb{Z}_{+}.
     $$
  Let $\Gamma_{\theta }\in \intoo{0,\infty}$ be such that 
$
\gamma(x)\leq \Gamma_{\theta }~\forall x\in \mathcal{A}+\theta \mathbb{B}_{n}
$. Consequently, we have
  \begin{align*}
\sum_{k=N}^{\infty}\Psi(\mathcal{R}(X,k))&\leq \sum_{k=0}^{\infty} \Gamma_{\theta}\overline{\alpha} \dha(\mathcal{R}(X,N+k),\mathcal{A})^{\bar{p}}\\
&\leq \Gamma_{\theta}\overline{\alpha}{M^{\bar{p}}}\sum_{k=0}^{\infty}(\lambda(\varepsilon,k))^{\bar{p}}<\infty.
   \end{align*}
Therefore, $\mathcal{V}(X)<\infty$. Let $\delta>0$ be such that 
$
X+\delta \mathbb{B}_{n}\subseteq \tilde{\mathcal{D}}_{\mathcal{A}}^{\mathcal{X}}.
$
Such $\delta$ exists due Lemma \ref{lem:delta-margin}. Therefore, for $Y\in \Comp$ with $\dh(X,Y)\leq \delta$, $Y\in \CompDOA$. Additionally, assume that $\delta$ satisfies, for all $Y\in \CompDOA$ with $\dh(X,Y)\leq \delta$, 
$
    \abs{\sum_{k=0}^{N-1}\Psi(\mathcal{R}(X,k))- \Psi(\mathcal{R}(Y,k))}\leq  \varepsilon,
    $
    and 
    $
    \dh(\mathcal{R}(X,N),\mathcal{R}(Y,N))\leq  {\varepsilon}/{2}.
    $
  This can be fulfilled due to the continuity of  $\mathcal{R}(\cdot,i),~i\in\intcc{1;N}$,  w.r.t. the Hausdorff distance  and the compactness of their images (Lemma \ref{lem:ReachableSetMap}),  the continuity of $\alpha$ and $\gamma$ over $\mathcal{X}$, and Lemma \ref{Lem:SupIsCts}.
Consequently, for all $Y\in \CompDOA$ with $\dh(X,Y)\leq \delta$, 
\begin{align*}
 \dha(\mathcal{R}(Y,N),\mathcal{A})\leq& \dha(\mathcal{R}(X,N),\mathcal{R}(Y,N))\\
 &+\dha(\mathcal{R}(X,N),\mathcal{A})\\
 \leq& \dh(\mathcal{R}(X,N),\mathcal{R}(Y,N))\\
 &+\dha(\mathcal{R}(X,N),\mathcal{A})  \leq \varepsilon.
\end{align*}
The $\ell_{p}$ stability indicates that 
   $
\dha(\mathcal{R}(Y,N+k),\mathcal{A})\leq M  \lambda(\varepsilon,k)\leq \theta ,~k\in \mathbb{Z}_{+}.
   $
   Therefore, 
   \begin{align*}
\sum_{k=N}^{\infty}\Psi(\mathcal{R}(Y,k))&\leq \sum_{k=0}^{\infty} \Gamma_{\theta}\overline{\alpha} \dha(\mathcal{R}(Y,N+k))^{\bar{p}}\\
&\leq \Gamma_{\theta}\overline{\alpha}{M^{\bar{p}}}\sum_{k=0}^{\infty}(\lambda(\varepsilon,k))^{\bar{p}}<\infty.
\end{align*}
Finally, for all $Y\in \CompDOA$ satisfying, $\dh(X,Y)\leq \delta$,
\begin{align*}      \abs{\mathcal{V}(X)-\mathcal{V}(Y)} \leq&
 \abs{\sum_{k=0}^{N-1}\Psi(\mathcal{R}(X,k))- \Psi(\mathcal{R}(Y,k))}\\
       &+\sum_{k=N}^{\infty}\Psi(\mathcal{R}(X,k))+\sum_{k=N}^{\infty}\Psi(\mathcal{R}(Y,k))\\
       &\leq \varepsilon+ 2\Gamma_{\theta}\overline{\alpha}{M^{\bar{p}}}\sum_{k=0}^{\infty}(\lambda(\varepsilon,k))^{\bar{p}}
   .
   \end{align*}
Since $\varepsilon>0$ is arbitrary, and the function 
$
s \longmapsto \sum_{k=0}^{\infty} \lambda(s,k)^{\bar p}
$
is continuous with 
\(
\sum_{k=0}^{\infty} \lambda(0,k)^{\bar p}=0,
\)
it follows that the right-hand side of the above inequality can be made arbitrarily small by choosing $\varepsilon>0$ sufficiently small. This completes the proof.

\end{proof}
Next, we want to illustrate the continuity of $\mathcal{W}$ over $\Comp$. This will require analyzing the behavior of $\mathcal{V}$ at the boundary of $\mathcal{D}^{\mathcal{X}}_{\mathcal{A}}$. In the next result, we show that $\mathcal{V}$ blows up as its argument gets closer to the boundary (in some set-valued sense).
\begin{theorem}\label{thm:DOABoundary} Let $X\in \Comp$  satisfy $X\cap \partial \tilde{\mathcal{D}}_{\mathcal{A}}^{\mathcal{X}}\neq \emptyset$ and $\{X_{k}\}_{k\in \mathbb{Z}_{+}}$ be a sequence of compact sets converging to $X$ in the sense of the Hausdorff distance. Then  $\lim_{k\rightarrow \infty}\mathcal{V}(X_{k})= \infty$.
\end{theorem}
\begin{proof}
    Without loss of generality, assume $X_{k}\in \CompDOA$ for all $k\in \mathbb{Z}_{+}$. Recall the definitions of $M$, $r$ and $\lambda$ in Assumption \ref{Assumptions}. Let $\theta \in \intoo{0,\infty}$ be such that
$
\theta<r$ and 
$\mathcal{A}+\theta\mathbb{B}_{n}\subseteq \mathcal{X}.
    $
    Let $\tilde{r}\in \intoo{0, \theta/M^{2}}   $.
    For each $k\in \mathbb{Z}_{+}$, let $T_{k}\in \mathbb{Z}_{+}$ be the first time instance such that 
    $
\mathcal{R}(X_{k},T_{k})\subseteq  \mathcal{A}+\tilde{r}\mathbb{B}_{n}.
    $
Note that, for all $j\in \intcc{0;k-1}$,
    $
\dha(\mathcal{R}(X,j),\mathcal{A})\geq \tilde{r}.
    $
    If the sequence $\{T_{k}\}_{k\in \mathbb{Z}_{+}}$ diverges to $\infty$, then we have
    $$
\mathcal{V}(X_{k})\geq \sum_{j=0}^{T_{k}-1}\Psi(\mathcal{R}(X_{k},j))\geq \underline{\gamma}\underline{\alpha}\tilde{r}^{\bar{p}}(T_{k}-1),~k\in \mathbb{Z}_{+}, 
    $$
    and, consequently,
     $\lim_{k\rightarrow \infty}\mathcal{V}(X_{k})=\infty$.
    Assume that $\{T_{k}\}_{k\in \mathbb{Z}_{+}}$ does not diverge to $\infty$. Then there exists a bounded subsequence, again denoted $\{T_{k}\}_{k\in \mathbb{Z}_{+}}$, with an upper bound $T\in \mathbb{Z}_{+}$ such that  
    $
T_{k}\leq T,~k\in \mathbb{Z}_{+}.
    $
As $
\tilde{r}\leq {\theta}/{M^{2}}\leq \theta<r,$
it follows that  
$
\mathcal{R}(X_{k},T_{k})\subseteq \mathcal{A}+\tilde{r} \mathbb{B}_{n}\subseteq \mathcal{A}+ r \mathbb{B}_{n},~k\in \mathbb{Z}_{+}.
$
The $\ell_{p}$ stability indicates that 
$$
\dha(\mathcal{R}(X_{k},T_{k}+j),\mathcal{A})\leq M\lambda(\tilde{r},j)\leq {\theta}/{M},~j,k\in \mathbb{Z}_{+}. 
$$
Hence, 
$
\mathcal{R}(X_{k},T)\subseteq \mathcal{A}+  (\theta/M) \mathbb{B}_{n}\subseteq r\mathbb{B}_{n},~k\in \mathbb{Z}_{+},
$
implying, using the continuity of $\mathcal{R}(\cdot,T)$ w.r.t. $\dh$, that 
$
\mathcal{R}(X,T)\subseteq \mathcal{A}+  (\theta/M) \mathbb{B}_{n}\subseteq \mathcal{A}+r \mathbb{B}_{n}.
$
Therefore, 
$$
\dha(\mathcal{R}(X,T+j),\mathcal{A})\leq  M\lambda((\theta/M),j)\leq \theta,~j\in \mathbb{Z}_{+}.
$$
This implies that 
$
\mathcal{R}(X,T+j)\subseteq  \mathcal{X},~j\in \mathbb{Z}_{+},
$
and 
$
\lim_{j\rightarrow \infty}\dha(\mathcal{R}(X,j),\mathcal{A})=0.
$
If $
\mathcal{R}(X,j)\subseteq  \mathcal{X}~\forall j\in \intcc{0;T-1},
$
it follows that $\mathcal{V}(X)<\infty$, implying $X$ is a proper subset of $\tilde{\mathcal{D}}_{\mathcal{A}}^{\mathcal{X}}$, which yields a contradiction. Now, assume $
\mathcal{R}(X,j)\cap \mathbb{R}^{n}\setminus\mathcal{X}\neq \emptyset
$ for some $j\in \intcc{0;T-1}$. Then, using the continuity of $\mathcal{R}(\cdot,j)$ w.r.t. $\dh$, it follows that $\lim_{k\rightarrow \infty}\dh(\mathcal{R}(X_{k},j),\mathcal{R}(X,j))=0
$.  
This yields $
\lim_{k\rightarrow \infty}\sup_{y\in \mathcal{R}(X_{k},j)}\gamma(y)= \infty
$ and consequently $\lim_{k\rightarrow \infty}\mathcal{V}(X_{k})=\infty$.
\end{proof}
Using Theorems \ref{thm:VCts} and \ref{thm:DOABoundary}, we have:
\begin{corollary}
   $\mathcal{W}$ is continuous over  $\Comp$. 
\end{corollary}

\subsection{Characterizing the value functions: Bellman-type equations}
\label{sec:Zubov}
In this section, we derive Bellman-type equations corresponding to the functions $\mathcal{V}$ and $\mathcal{W}$, respectively.

\begin{theorem}\label{Thm:LyapunovEquation}
 $\mathcal{V}$ satisfies the  equation (w.r.t. to the function $v$)
\begin{align}\label{eq:LyapunovEqn}
v(X)
= \Psi(X)+v(F(X)),~X\in \Comp.
\end{align}
\end{theorem}

\begin{proof}
    Given $X\in \Comp$ and using definition of $\mathcal{V}$ and Lemma \ref{lem:SemiGroup}, we have 
    \begin{align*}\mathcal{V}(X)&= \sum_{k=0}^{\infty}\Psi(\mathcal{R}(X,k))\\
    &=\Psi(\mathcal{R}(X,0))+\sum_{k=1}^{\infty}\Psi(\mathcal{R}(X,k))\\
    &=\Psi(X)+\sum_{k=0}^{\infty}\Psi(\mathcal{R}(X,k+1))\\
    &=\Psi(X)+\mathcal{V}(F(X)).
    \end{align*}
   Note that the above decomposition is valid even if $\mathcal{V}(X)$ is infinite.
\end{proof}
\begin{theorem}\label{thm:ZubovEqn1}
 $\mathcal{W}$ satisfies the  equation (w.r.t. to the function $w$)
\begin{equation}\label{eq:ZubovEqn1}
w(X)-w({F}(X))=\xi(X)(1-w(F(X))),~X\in \Comp,
\end{equation}
where 
\begin{equation}\label{eq:xi}
\xi(X)\defas 1-\exp(-\Psi(X)).
\end{equation}
\end{theorem}

\begin{proof}
Using Theorem \ref{Thm:LyapunovEquation}, we have, for any $X\in \Comp$, 
\begin{align*}
\mathcal{W}(X)&= 1-\exp(-\mathcal{V}(X))\\
&= 1-\exp(-\mathcal{V}({F}(x))-\Psi(X))\\
&=1-\exp(-\Psi(X))(1-\mathcal{W}({F}(X))),
\end{align*}
implying  
\begin{align*}
\mathcal{W}(X)-\mathcal{W}({F}(X))=&1-\mathcal{W}({F}(X))\\
     &-\exp(-\Psi(X))(1-\mathcal{W}({F}(x)))\\
=&\xi(X)(1-\mathcal{W}({F}(X))).
\end{align*} 
    
\end{proof}
We have the following version of the Bellman-type equation for the function $\mathcal{W}$ inside the safe set $\mathcal{X}$.
\begin{theorem}\label{thm:ZubovEqn2}
If $w$ satisfies equation \eqref{eq:ZubovEqn1} over $\mathcal{K}(D)$ for some nonempty  $D\subseteq \mathcal{X}$, then for  all $X\in \mathcal{K}(D)$, 
$w$ satisfies  the equation
\begin{equation}\label{eq:ZubovEqn2}
w(X)-w({F}(X))=\beta(X)(1-w(X)),
\end{equation}
where 
\begin{equation}\label{eq:beta}
\beta(X)\defas \exp(\Psi(X))-1.
\end{equation}
\end{theorem}
\begin{proof}
    When $X\in \mathcal{K}(D)$,  $\sup_{y\in X}\gamma(y)<\infty$, implying $\Psi(X)<\infty$. Therefore, using equation \eqref{eq:ZubovEqn1}, 
    $
1-{w}({F}(X))=\exp(\Psi(X))(1-w(X)).
    $
    Hence,
   \begin{align*}
       w(X)- w({F}(X))&=\exp(\Psi(X))(1-w(X))-1+w(X)\\
       &= \exp(\Psi(X))(1-w(X))-(1-w(X))\\
       &= \beta(X)(1-w(X)).
       \end{align*}
\end{proof}
We have shown that the value functions $\mathcal{V}$ and $\mathcal{W}$ are solutions to the equations \eqref{eq:LyapunovEqn} and \eqref{eq:ZubovEqn1}, respectively. Next, we show that the solutions to these equations are unique with respect to functions that satisfy a particular limit condition concerning reachable sets converging to the RIS $\mathcal{A}$. We start with the following technical result:
\begin{lemma}\label{lem:boundedness}
    Assume that $w\colon \CompDOA\rightarrow \mathbb{R}$   satisfies equation \eqref{eq:ZubovEqn1} over $\CompDOA$ and for any  compact set $X$,
    $
     \lim_{k\rightarrow \infty} \dha(\mathcal{R}(X,k),\mathcal{A})=0\Rightarrow \lim_{k\rightarrow \infty} w(\mathcal{R}(X,k))=0
    $. Then, $w(X)<1$ for all $X\in \CompDOA$.
\end{lemma}
\begin{proof}
  Using Theorem \ref{thm:ZubovEqn2}, $w$ satisfies equation \eqref{eq:ZubovEqn2} over $\CompDOA$. Assume that for some  $X\in \CompDOA$, $w(X)\geq 1$. We have 
  $\mathcal{R}(X,k)\subseteq  \tilde{\mathcal{D}}_{\mathcal{A}}^{\mathcal{X}}~\forall k\in \mathbb{Z}_{+},
  $ and $\lim_{k\rightarrow \infty}\dha(\mathcal{R}(X,k),\mathcal{A})=0$. Using equation \eqref{eq:ZubovEqn2}, we have 
  $
w(X)-w(\mathcal{R}(X,1))=\beta(X)(1-w(X))\leq 0$. Hence, $
w(\mathcal{R}(X,1))\geq w(X)\geq 1,$
and by induction, we have $
w(\mathcal{R}(X,k+1))\geq w(\mathcal{R}(X,k))\geq 1
$
for all $k\in \mathbb{Z}_{+}$. Hence, $
\lim_{k\rightarrow\infty}w(\mathcal{R}(X,k))\neq 0,
$
which yields a contradiction as $\lim_{k\rightarrow \infty}w(\mathcal{R}(X,k))$ 
must be zero due to the assumption on $w$, and that completes the proof.
\end{proof}

Now we show the uniqueness result for the Bellman-type equation associated with the function $\mathcal{V}_{\mathcal{X}}$.
\begin{theorem}\label{thm:LyapunovEqnUniqueness}
    Let $\mathbf{v}\colon \CompDOA\rightarrow \mathbb{R}$   satisfy equation \eqref{eq:LyapunovEqn} over   
    $\CompDOA$  and for any  compact  $X$
    $
    \lim_{k\rightarrow \infty} \dha(\mathcal{R}(X,k),\mathcal{A})=0\Rightarrow \lim_{k\rightarrow \infty} \mathbf{v}(\mathcal{R}(X,k))=0
    $. Then,
    $
\mathbf{v}(X)=\mathcal{V}(X)$ for all $X\in \CompDOA.
    $
\end{theorem}
\begin{proof}
 Let $X\in \CompDOA$. We have, using Lemmas \ref{lem:ReachableSetMap} and \ref{lem:compDOA_inv}, 
    $\mathcal{R}(X,k)\subseteq \tilde{\mathcal{D}}_{\mathcal{A}}^{\mathcal{X}}~\forall k\in \mathbb{Z}_{+}.
    $
 Let $k\in \mathbb{N}$. We consequently have 
 $\mathbf{v}(X)-\mathbf{v}(\mathcal{R}(X,k))=\sum_{j=0}^{k-1}(\mathbf{v}(\mathcal{R}(X,j))-\mathbf{v}(\mathcal{R}(X,j+1))) =\sum_{j=0}^{k-1}\Psi(\mathcal{R}(X,j))$.
As $\lim_{k\rightarrow \infty}\dha(\mathcal{R}(X,k),\mathcal{A})=0$, it follows by assumption that $\lim_{k\rightarrow \infty}\mathbf{v}(\mathcal{R}(X,k))=0$. We also have  $\mathcal{V}(X)<\infty$ ($\sum_{j=0}^{k-1}\Psi(\mathcal{R}(X,j))$ converges as $X \in \CompDOA$), then taking the limit as $k\rightarrow \infty$ in the both sides of the above equation results in
    $
     \mathbf{v}(X)=\sum_{j=0}^{\infty}\Psi(\mathcal{R}(X,j))=\mathcal{V}(X).
    $
\end{proof}

Now, we introduce the uniqueness result for the equation associated with the value function $\mathcal{W}_{\mathcal{X}}$.
\begin{theorem}
    Let $\mathbf{w}\colon \Comp\rightarrow \mathbb{R}$ be a bounded function  satisfying 
    equation \eqref{eq:ZubovEqn1} over $\Comp$ and for any  compact set $X$,   $
  \lim_{k\rightarrow \infty} \dha(\mathcal{R}(X,k),\mathcal{A})=0\Rightarrow \lim_{k\rightarrow \infty} \mathbf{w}(\mathcal{R}(X,k))=0
    $.  
Then, $\mathbf{w}=\mathcal{W}$.
\end{theorem}
\begin{proof}
   First note that the difference of two solutions, $\mathbf{w}_{1}$ and $\mathbf{w}_{2}$, to equation \eqref{eq:ZubovEqn1} satisfies, for $X\in \Comp$,
\begin{equation}\label{eq:difference1}
    \mathbf{w}_{1}(X)-\mathbf{w}_{2}(X)=(1-\xi(X))(\mathbf{w}_{1}(F(X))-\mathbf{w}_{2}(F(X))).
\end{equation}
If $X\in \Comp$ satisfies $X\cap (\mathbb{R}^{n}\setminus \mathcal{X})\neq \emptyset$, then we have $\xi(X)=1$, implying, using equation \eqref{eq:ZubovEqn1}, $\mathbf{w}(X)=1=\mathcal{W}(X)$. 

Over  $\CompDOA$, $\mathbf{v}(\cdot)\defas -\ln(1-\mathbf{w}(\cdot))$ is well-defined due to Lemma \ref{lem:boundedness}, satisfying equation \eqref{eq:LyapunovEqn}. Then, using Theorem \ref{thm:LyapunovEqnUniqueness},  it follows that $\mathbf{v}(X)=\mathcal{V}(X)$. Hence $\mathbf{w}(X)=\mathcal{W}(X)$ for all   $X\in \CompDOA$.

Let $X\in \mathcal{K}(\mathcal{X})$ be such that $\mathcal{R}(X,j)\cap \mathbb{R}^{n}\setminus \mathcal{X}\neq \emptyset
$
or
$\mathcal{R}(X,j)\subseteq \mathcal{ D}_{\mathcal{A}}^{\mathcal{X}}
$
for some $j\in \mathbb{N}$ and $
\mathcal{R}(X,k)\subseteq \mathcal{X}~\forall k\in \intcc{0;j-1}.
$
Define
$
\Delta_{w}(\cdot)\defas \mathbf{w}(\cdot)-\mathcal{W}(\cdot).
$
Then, using equation \eqref{eq:difference1},  
\begin{align*}
\Delta_{w}(\mathcal{R}(X,j-1))&=(1-\xi(\mathcal{R}(X,j-1)))(\Delta_{w}(\mathcal{R}(X,j)))=0,
\end{align*}
as $\Delta_{w}(\mathcal{R}(X,j))=0$, and by an inductive argument, we have 
$
\Delta_{w}(\mathcal{R}(X,k))=0~\forall k\in \intcc{0;j},
$
implying $\mathbf{w}(X)=\mathcal{W}(X)$. 

 Now, let $X\in\mathcal{K}( \mathcal{X})$ be such that $\mathcal{R}(X,k)\subseteq  \mathcal{X}$
 for all $k\in \mathbb{Z}_{+}$ and $\lim_{k\rightarrow \infty}\dha(\mathcal{R}(X,k),\mathcal{A})\neq 0$. Obviously, 
 $
 \mathcal{R}(X,k)\not\subseteq\tilde{\mathcal{D}}_{\mathcal{A}}^{\mathcal{X}}
 $
 for all $k\in \mathbb{Z}_{+}$.  
Then it follows from $\mathcal{A}\in \CompDOA$ that there exists $\theta>0$ such that $\dha(\mathcal{R}(X,k),\mathcal{A})\geq \theta~\forall k\in \mathbb{Z}_{+}$. Assume $\mathbf{w}(X)\neq \mathcal{W}(X)$. Note the difference of two solutions, $\mathbf{w}_{1}$ and $\mathbf{w}_{2}$, to equation \eqref{eq:ZubovEqn1} over $\mathcal{X}$ (which are also solutions to \eqref{eq:ZubovEqn2} over $\mathcal{X}$ using Theorem \ref{thm:ZubovEqn2}) satisfies, for $Y\in \mathcal{K}(\mathcal{X})$,
\begin{equation}\label{eq:difference2}
    \mathbf{w}_{1}(F(Y))-\mathbf{w}_{2}(F(Y))=(1+\beta(Y))(\mathbf{w}_{1}(Y)-\mathbf{w}_{2}(Y)).
\end{equation} 

It then follows that, for all $j\in \mathbb{N}$,   
$$
\Delta_{w}(\mathcal{R}(X,j))=\prod_{k=0}^{j-1}(1+\beta(\mathcal{R}(X,k)))\Delta_{w}(X). 
$$
For all $k \in \intcc{0;j-1}$, we have 
\begin{align*}
1+\beta(\mathcal{R}(X,k))=&\exp(\Psi(\mathcal{R}(X,k)))\\
&\geq \exp(\underline{\gamma}\underline{\alpha}\dha(\mathcal{R}(X,k),\mathcal{A})^{\bar{p}})\\
&\geq  \exp(\underline{\gamma}\underline{\alpha}\theta^{\bar{p}}).
\end{align*}
Hence,
  $$
\abs{\Delta_{w}(\mathcal{R}(X,j))}\geq \exp(j\underline{\gamma}\underline{\alpha}\theta^{\bar{p}})\abs{\Delta_{w}(X)}.
$$
As $\lim_{ j\rightarrow\infty}\exp(j\underline{\gamma}\underline{\alpha}\theta^{\bar{p}})=\infty$, it then follows that for each $M>0$, there exists a compact $Z\subseteq  \mathcal{X}$ (which corresponds to the image of $\mathcal{R}(X,\cdot)$ at some time step) such that  $\mathbf{w}(Z)>M+\mathcal{W}(Z)\geq M-1
$ or 
$\mathbf{w}(Z)<-M+\mathcal{W}(Z)\leq 
-(M-1),
$
which is equivalent to $\abs{w(Z)}> M-1$. Hence, $\mathbf{w}$ is unbounded over $\mathcal{K}(\mathcal{X})$, a contradiction.
\end{proof}
We conclude this section with the following important property:
\begin{theorem}\label{eq:ValueFunctionsAreLyapunov}
The functions $\mathcal{V}$ and
$\mathcal{W}$ are robust  Lyapunov functions in the following sense: they satisfy the conditions
\begin{align*}
  \mathcal{V}(\{x\}),\mathcal{W}(\{x\})&>0 ~\forall x\not\in \mathcal{A},\\ 
\mathcal{V}(\{x\})=\mathcal{W}(\{x\})&=0~\forall x\in \mathcal{A},\\
\mathcal{V}(\{x\})-\sup_{w\in W}\mathcal{V}(\{f(x,w)\})&>0 ~\forall x \in \tilde{\mathcal{D}}_{\mathcal{A}}^{\mathcal{X}}\setminus \mathcal{A}, \\
\mathcal{W}(\{x\})-\sup_{w\in W}\mathcal{W}(\{f(x,w)\})&>0~\forall x \in \tilde{\mathcal{D}}_{\mathcal{A}}^{\mathcal{X}}\setminus \mathcal{A}.
\end{align*}
\end{theorem}
\begin{proof}
The positive definiteness w.r.t. $\mathcal{A}$ follows immediately from the definitions. 

Let $x \in \tilde{\mathcal{D}}_{\mathcal{A}}^{\mathcal{X}}\setminus \mathcal{A}$.  using the monotonicity of of   $\mathcal{V}$ and $\mathcal{W}$ w.r.t. set inclusion (Lemma \ref{lem:Monotonicity}), equation \eqref{eq:LyapunovEqn}, the positive definiteness of $\Psi(\{\cdot\})$ w.r.t. $\mathcal{A}$, and the definition of $\mathcal{W}$, 
\begin{align*}
\mathcal{V}(\{x\})-\sup_{w\in W}\mathcal{V}(\{f(x,w)\}) \geq& \mathcal{V}(\{x\})-\mathcal{V}(F(\{x\}))\\
&=\Psi(\{x\})>0,
\end{align*}
and 
\begin{align*}
\mathcal{W}(\{x\})-\sup_{w\in W}\mathcal{W}(\{f(x,w)\}) \geq& \mathcal{W}(\{x\})-\mathcal{W}(F(\{x\}))\\
&=1-\exp(-\mathcal{V}(\{x\}))\\
&-(1-\exp(-(\mathcal{V}(\{x\})\\
&-\Psi(\{x\})))>0.
\end{align*}
\end{proof}

\section{DOA estimation}\label{sec:VFApproximation}
In the previous sections, we introduced novel value functions defined on metric spaces of compact sets, which characterize the safe robust DOA. 
We now aim to employ these value functions to derive constructive and practically meaningful estimates of the safe robust DOA via means of NN learning.

\subsection{Finite-horizon approximation of the value functions}
\label{sec:ReachSetApproximation}

Within the proposed learning framework, it is essential to evaluate the value functions 
$\mathcal{V}$ and $\mathcal{W}$ over singleton sets, at least approximately. 
Let $N_{s} \in \mathbb{N}$ be a user-defined truncation horizon specifying the maximum number of steps. 
We introduce the finite-horizon approximation
\[
\mathcal{V}_{N_{s}}(X) \defas 
\sum_{k=0}^{N_{s}} \Psi\big(\mathcal{R}(X,k)\big),
\qquad X \in \Comp.
\]

Since the reachable sets $\mathcal{R}(X,k)$, $k\in \intcc{0;N_{s}}$, cannot in general be computed exactly, 
we employ approximate representations. 
Specifically, we approximate the family $\{\mathcal{R}(X,k)\}_{k\in \intcc{0;N_{s}}}$ by a sequence of sets 
$\{\tilde{\mathcal{R}}(X,k)\}_{k\in \intcc{0;N_{s}}}$ satisfying 
$\mathcal{R}(X,k) \approx \tilde{\mathcal{R}}(X,k)$ for all $k \in \intcc{0;N_{s}}$. 

The literature provides a wide range of approximation schemes for discrete-time systems 
(see, e.g., \cite{kuhn1998rigorously,kieffer2002guaranteed,alamo2005guaranteed,yang2020accurate} 
and the references therein). 
In general, higher approximation accuracy is achieved at the expense of increased computational complexity. 
Moreover, the chosen approximation must allow for tractable evaluation of the function 
$\Psi$ defined in \eqref{eq:Psi}. 
Again, in our setting it suffices to approximate reachable sets 
for singleton initial sets of the form $X=\{x\}$ with $x \in \mathbb{R}^{n}$. 

In Section~\ref{sec:NumericalExamples}, we adopt a trajectory-based approximation 
that balances accuracy and computational tractability. 
Let $N_{\mathrm{traj}} \in \mathbb{N}$ denote the number of trajectories used to approximate the reachable sets, 
and let $\pi^{(j)} \colon \mathbb{Z}_{+} \to W$, $j\in \intcc{1;N_{\mathrm{traj}}}$, 
be randomly generated disturbance signals. 
We define
\[
\tilde{\mathcal{R}}(\{x\},k)
=
\bigcup_{j=1}^{N_{\mathrm{traj}}}
\left\{ \varphi_{x}^{\pi^{(j)}}(k) \right\},
\]
that is, the reachable set at time $k$ is approximated by the finite collection of trajectory values induced by the disturbance sequences $\pi^{(j)}$. 
This construction makes the evaluation of $\Psi$ significantly more tractable in practice. 
The accuracy of the approximation depends on using a sufficiently large number of disturbance realizations, i.e., $N_{\mathrm{traj}}$ should be chosen large enough.

Based on the reachable-set approximations, we define $\tilde{\mathcal{V}}_{N_{s}}(X) \defas 
\sum_{k=0}^{N_{s}} \Psi\big(\tilde{\mathcal{R}}(X,k)\big)$,~ $X \in \Comp$. Finally, we approximate $\mathcal{W}$ by
$
\tilde{\mathcal{W}}_{N_{s}}(X)
\defas 
1 - \exp\!\big(-\tilde{\mathcal{V}}_{N_{s}}(X)\big)$,
$X \in \Comp$.

\subsection{Physics-informed and set-based neural network learning}
\label{sec:PINNLearning}
Since the value function $\mathcal{W}$ 
is a robust Lyapunov function when restricted to singleton sets 
(Theorem~\ref{eq:ValueFunctionsAreLyapunov}), 
we approximate $\mathcal{W}(\{x\})$ for $x \in \mathbb{R}^n$. 
The resulting approximation serves as a robust Lyapunov function candidate, 
which can subsequently be used to estimate the safe robust domain of attraction (DOA).

In this paper, we propose a physics-informed learning framework 
that incorporates the set-based Bellman (Zubov-type) equation~\eqref{eq:ZubovEqn1}. 
A principal challenge in incorporating \eqref{eq:ZubovEqn1} into the training procedure 
is the set-valued nature of the term $F(\{x\})$, 
which arises due to the presence of disturbance.

We therefore impose the following structural assumption.

\begin{assumption}\label{assumption:embedding}
Define
\[
\mathcal{S}
:=
\big\{ \{x\} : x \in \mathbb{R}^n \big\}
\;\bigcup\;
\big\{ F(\{x\}) : x \in \mathbb{R}^n \big\}.
\]
Assume there exists a closed-form embedding
\[
\mathcal{T} : \mathcal{S} \to \mathbb{R}^L,
\]
for some fixed integer $L \in \mathbb{N}$, 
such that $\mathcal{T}$ is injective on $\mathcal{S}$. 
That is, each singleton set $\{x\}$ and each reachable set $F(\{x\})$ 
admits a unique representation in $\mathbb{R}^L$.
\end{assumption}

\begin{remark}
Assumption~\ref{assumption:embedding} is satisfied in several important settings.

The simplest case arises when the disturbance set is a singleton. 
In this situation, $F(\{x\})$ is itself a singleton, and the embedding reduces to 
$\mathcal{T}(\{x\}) = x$.

More generally, the assumption holds whenever there exists a parameterized class of sets 
(e.g., polytopes, polynomial zonotopes, or fixed-template sublevel sets) 
with a fixed number of parameters $L$ 
that can represent both $\{x\}$ and $F(\{x\})$.

For instance, if $F(\{x\})$ is a polytope described by a fixed number of linear inequalities, 
then $\mathcal{T}$ may be defined using the corresponding inequality coefficients.

A particularly relevant case occurs when the disturbance set $W$ 
is a hyper-interval set and distinct disturbance components 
affect different components of the dynamics. 
In this situation, $F(\{x\})$ is itself a hyper-interval, and one may define
\[
\mathcal{T}(\Hintcc{a,b})
=
\begin{bmatrix}
\mathrm{center}(\Hintcc{a,b}) \\
\mathrm{radius}(\Hintcc{a,b})
\end{bmatrix},
\]
where $\Hintcc{a,b}$ denotes a closed interval.

The assumption is also satisfied when the system dynamics are affine in the disturbance 
for each fixed $x$, and the disturbance set $W$ is polytopic.

If the assumption does not hold, one may replace $F$ by a continuous outer approximation 
$\tilde{F}$ that satisfies the embedding condition and fulfills
\[
F(\{x\}) \subseteq \tilde{F}(\{x\}), 
\quad x \in \mathbb{R}^{n}.
\]
While this ensures applicability of the framework, it may introduce additional conservatism in the resulting DOA estimates.
\end{remark}

Let
\[
\tilde{\omega}_{\mathrm{nn},\theta}(\cdot) : \mathbb{R}^L \to \mathbb{R}
\]
be a feedforward neural network with a fixed architecture 
(i.e., a fixed number of hidden layers and neurons), 
whose weight matrices and bias vectors are collected in the parameter vector $\theta$. 
We determine $\theta$ so that 
\[
\tilde{\omega}_{\mathrm{nn},\theta}(\mathcal{T}(\cdot))
\]
approximates the value function over a compact learning domain 
$\mathbb{X}_{\ell} \subset \mathbb{R}^n$ satisfying
\[
\mathcal{X} \cap \mathbb{X}_{\ell} \neq \emptyset,
\qquad
\mathcal{A} \subseteq \mathbb{X}_{\ell}.
\]

The training problem consists of minimizing the composite loss functional
\begin{equation}
\label{eq:lossV}
\mathcal{L}(\theta)
=
\lambda_{\mathrm{d}} \mathcal{L}_{\mathrm{d}}(\theta)
+
\lambda_{\mathrm{pi}} \mathcal{L}_{\mathrm{pi}}(\theta),
\end{equation}
where $\lambda_{\mathrm{d}}, \lambda_{\mathrm{pi}} > 0$ are weighting coefficients and
$$
\mathcal{L}_{\mathrm{d}}(\theta)
\defas
\frac{1}{N_{\mathrm{d}}}
\sum_{i=1}^{N_{\mathrm{d}}}
J_{\mathrm{d},\theta}\bigl(\{z_{\mathrm{d}}^{(i)}\}\bigr),
$$

$$
\mathcal{L}_{\mathrm{pi}}(\theta)
\defas
\frac{1}{N_{\mathrm{pi}}}
\sum_{i=1}^{N_{\mathrm{pi}}}
J_{\mathrm{pi},\theta}\bigl(\{z_{\mathrm{pi}}^{(i)}\}\bigr).
$$
The data-driven residual is defined as
\[
J_{\mathrm{d},\theta}(\{x\})
=
\left(
\tilde{\mathcal{W}}_{N_s}(\{x\})
-
\tilde{\omega}_{\mathrm{nn},\theta}\bigl(\mathcal{T}(\{x\})\bigr)
\right)^2,
\]
which penalizes the discrepancy between the neural network output 
and the $N_s$-step approximation 
$\tilde{\mathcal{W}}_{N_s}$ 
of the ground-truth value function 
$\mathcal{W}$. 

The physics-informed residual is given by
\begin{align*}
J_{\mathrm{pi},\theta}(\{x\})
=&
\Big(
\tilde{\omega}_{\mathrm{nn},\theta}\bigl(\mathcal{T}(\{x\})\bigr)
-
\tilde{\omega}_{\mathrm{nn},\theta}\bigl(\mathcal{T}(F(\{x\}))\bigr)\\
&-
\xi(\{x\})
\bigl(
1
-
\tilde{\omega}_{\mathrm{nn},\theta}\bigl(\mathcal{T}(F(\{x\}))\bigr)
\bigr)
\Big)^2,
\end{align*}
which penalizes violations of the Bellman-type functional equation~\eqref{eq:ZubovEqn1} 
in a residual-minimization sense.

The collocation points 
$\{z_{\mathrm{d}}^{(i)}\}_{i=1}^{N_{\mathrm{d}}} \subset \mathbb{X}_{\ell}$ 
and 
$\{z_{\mathrm{pi}}^{(i)}\}_{i=1}^{N_{\mathrm{pi}}} \subset \mathbb{X}_{\ell}$ 
are independently sampled from $\mathbb{X}_{\ell}$. 
The former are used to fit the approximate ground-truth values 
$\tilde{\mathcal{W}}_{N_s}$, 
while the latter enforce the governing functional equation 
in a physics-informed manner.

Finally, we define the state-space function
\[
\omega_{\mathrm{nn},\theta}(x)
\defas
\tilde{\omega}_{\mathrm{nn},\theta}\bigl(\mathcal{T}(\{x\})\bigr),
\qquad x \in \mathbb{R}^n.
\]
That is, $\omega_{\mathrm{nn},\theta}$ denotes the composition of the embedding 
$\mathcal{T}$ with the trained neural network. 
This function will be regarded as the robust neural Lyapunov function candidate 
and will be employed for safe robust DOA estimation.
\subsection{Verified safe robust ROAs from neural network approximation}
\label{sec:NNVerification}

We now aim to construct a relatively large safe robust ROA 
by leveraging the learned neural network approximation $\omega_{\mathrm{nn}}$.  The proposed verification schemes extend those developed in \cite{serry2025safe}, while accounting for uncertainties, and are conceptually aligned with the approach presented in \cite{lu2024estimating}, while additionally considering state constraints.
For notational simplicity, the parameter dependence on $\theta$ 
is henceforth omitted from $\omega_{\mathrm{nn}}$.
\begin{remark}
The verification framework presented herein is applicable to any function $\omega$ in place of the NN-based Lyapunov function $\omega_{\mathrm{nn}}$. 
Such functions $\omega$ may be obtained through standard Lyapunov analysis or SOS optimization techniques.
\end{remark}
To this end, let   $\mathbb{X}_{v}\subseteq \mathbb{R}^{n}$ be a compact verification domain such that $$ \mathcal{X}\cap  \mathbb{X}_{v}\neq \emptyset,~\mathcal{A}\subseteq \mathbb{X}_{v}.
$$

\subsubsection{Initial ROA}
 DOA verification and estimation typically starts with a small ROA, where we assume we are provided with a parameter $c_{1}\in \mathbb{R}$ and a set $\mathbb{E}_{c_{1}}\subseteq \mathbb{R}^{n}$ satisfying:
 \begin{assumption}\label{assumption:InitialROA}
     The set $\mathbb{E}_{c_{1}}\subseteq \mathbb{R}^{n}$ is a sublevel set of the form:
     $$
\mathbb{E}_{c_{1}}=\{x\in \mathbb{R}^{n}: \nu(x)\leq c_{1}\},
     $$
     where $\nu \colon \mathbb{R}^{n}\rightarrow \mathbb{R}$. Moreover, $\mathbb{E}_{c_{1}}$ is   a safe robust ROA within $\mathcal{X}\cap \mathbb{X}_{v}$. 
 \end{assumption}

 In  Appendix \ref{sec:Ellipsoidal ROAs}, we provide a systematic way to obtain such an ROA for exponentially stable EPs under additional smoothness assumptions on $f$.
\begin{remark}
Obtaining initial ROA estimates for a given RIS is, in general, a highly nontrivial task. For exponentially stable equilibrium points with sufficiently smooth right-hand sides, and for certain polynomially stable equilibrium points, standard Lyapunov analysis can be employed to derive initial estimates (see, e.g., Appendices~\ref{sec:Ellipsoidal ROAs} and~\ref{sec:LyapunovAnalysisPolynomial}). However, for general nonlinear systems and nonsingleton RISs that do not contain attracting equilibrium points, such estimates are significantly more difficult to construct. Developing systematic methods for these broader settings remains an open problem and an active area of research.
\end{remark}

\subsubsection{Enlarging the initial ROA through verification} 
\label{sec:EnlargingEllipse}
We  can enlarge the provided initial safe ROA via verification.
If   for some  $c_{2}\in \mathbb{R}$, with  $c_{2}>c_{1}$, and   some small positive constant $\varepsilon>0$,  we can verify the following: 
\begin{align}\nonumber
&((x,w)\in \mathbb{X}_{v}\times W)\wedge (c_{1} \le \nu(x) \le c_{2})\Rightarrow\\ &(\nu(f(x,w))-\nu(x)  \le -\varepsilon)
\wedge (g(x)<1) \wedge (f(x,w)\in \mathbb{X}_{v}),\label{eq:dVP}
\end{align}
 then:

\begin{lemma} \label{lemma:enlarge_quad}
 The set $$ \mathbb{E}_{c_{2}} := \{x\in \mathbb{X}_{v}:\,\nu(x)\le c_{2}\}
    $$
    is a safe robust ROA of \eqref{eq:System} within $\mathcal{X}\cap \mathbb{X}_{v}$. 
\end{lemma}

\begin{proof}
   Define
$
\mathbb{E}_{c_{2}/c_{1}}
\defas
\bigl\{ x \in \mathbb{X}_{v} \;:\; c_{1} < \nu(x) \le c_{2} \bigr\}.
$
By Assumption \ref{assumption:InitialROA},  $\mathbb{E}_{c_{1}}$ is a safe robust ROA within $\mathcal{X}\cap\mathbb{X}_{v}$.  We have $E_{c_{1}}=\{x\in \mathbb{X}_{v}:\,\nu(x)\le c_{2}\}$, which implies $\mathbb{E}_{c_{2}}=\mathbb{E}_{c_{2}/c_{1}}\bigcup \mathbb{E}_{c_1}$.  By condition \eqref{eq:dVP}, we  have that $\mathbb{E}_{c_{2}/c_{1}}\subseteq \mathcal{X}$ ($g_\mathcal{X}(\cdot)<1$). Therefore, and by definition, $\mathbb{E}_{c_{2}}=\mathbb{E}_{c_{1}}\bigcup\mathbb{E}_{c_{2}/c_{1}}\subseteq \mathcal{X}\cap \mathbb{X}_{v}$.    
Moreover, by  \eqref{eq:dVP}, for all 
$(x,w) \in \mathbb{E}_{c_{2}/c_{1}} \times W$,
$
\nu\bigl(f(x,w)\bigr)
\le
\nu(x) - \varepsilon
\le
c_{2} - \varepsilon
\le
c_{2}$ and $f(x,w)\in \mathbb{X}_{v}$, implying $f(x,w)\in \mathbb{E}_{c_{2}}$. Hence,
$
\mathbb{E}_{c_{2}}
=
\mathbb{E}_{c_{1}}
\bigcup
\mathbb{E}_{c_{2}/c_{1}}
$
is invariant. We now show that any trajectory starting in 
$\mathbb{E}_{c_{2}/c_{1}}$ enters $\mathbb{E}_{c_{1}}$ in finite time. 
Let $x \in \mathbb{E}_{c_{2}/c_{1}}$ and suppose, by contradiction, 
that there exists a disturbance sequence 
$\pi \in W^{\mathbb{Z}_{+}}$ such that
$
\varphi_{x}^{\pi}(k)
\in
\mathbb{E}_{c_{2}/c_{1}}$, $
\forall k \in \mathbb{Z}_{+}$. Then,
$
c_{1}
<
\nu\bigl(\varphi_{x}^{\pi}(k)\bigr)
\le
c_{2},
~
\forall k \in \mathbb{Z}_{+}$. Using condition~\eqref{eq:dVP} inductively, we obtain
$
\nu\bigl(\varphi_{x}^{\pi}(k)\bigr)
\le
\nu(x) - k \varepsilon
\le
c_{2} - k \varepsilon$, $\forall k \in \mathbb{Z}_{+}.
$
Choose any integer
$
k
>
(c_{2} - c_{1})/{\varepsilon}.
$
Then, $
\nu\bigl(\varphi_{x}^{\pi}(k)\bigr)
<
c_{1},
$
which contradicts the assumption that 
$\varphi_{x}^{\pi}(k) \in \mathbb{E}_{c_{2}/c_{1}}$. 
Therefore, every trajectory starting in 
$\mathbb{E}_{c_{2}/c_{1}}$ enters $\mathbb{E}_{c_{1}}$ in finite time. Finally, by Assumption \ref{assumption:InitialROA}, 
all trajectories in $\mathbb{E}_{c_{1}}$ converge to the RIS $\mathcal{A}$. 
Since every trajectory in $\mathbb{E}_{c_{2}}$ enters 
$\mathbb{E}_{c_{1}}$ in finite time, it follows that 
all trajectories in $\mathbb{E}_{c_{2}}$ converge to  $\mathcal{A}$.
\end{proof}

\subsubsection{NN  ROAs via verification} 
\label{sec:enlarge_neural}
The improved safe robust  ROA $\mathbb{E}_{c_{2}}$ can be utilized to produce an even potentially  larger safe robust ROA using the NN approximation $\omega_{\mathrm{nn}}$.  

If we can  verify  for some constants $\omega_{1}$ and $\omega_{2}$, with $\omega_{2}>\omega_{1}$, and some positive constant $\varepsilon>0$, the  following: 
\begin{align}\nonumber
((x, w)\in \mathbb{X}_{v}\times W)\wedge(\omega_{\mathrm{nn}}(x)\le \omega_{1}) \Rightarrow\\
 (\nu(x)\le c_{2}) \wedge( \omega_{\mathrm{nn}}(f(x,w))\le \omega_{2}),\label{eq:c2w2_inclusion}
\end{align}
\begin{align}
\nonumber
((x, w)\in \mathbb{X}_{v}\times W)\wedge(\omega_{1}\le \omega_{\mathrm{nn}}(x) \le \omega_{2})   \Rightarrow\\ \nonumber 
(\omega_{\mathrm{nn}}(f(x,w))-\omega_{\mathrm{nn}}(x)  \le -\varepsilon) \wedge\\ 
(g(x) < 1)\wedge(f(x,w) \in \mathbb{X}_{v})\label{eq:dW},
\end{align}
 then the following holds.

\begin{lemma}
    Suppose that \eqref{eq:c2w2_inclusion} and \eqref{eq:dW} and the conditions in Lemma \ref{lemma:enlarge_quad} hold. Then the set $$\mathbb{W}_{\omega_{2}} := \{x\in \mathbb{X}_{v}:\,\omega_{\mathrm{nn}}(x)\le \omega_{2}\}$$
    is a safe robust  ROA of \eqref{eq:System} within $\mathcal{X}\cap \mathbb{X}_{v}$. 
\end{lemma}

\begin{proof}
    Define $\mathbb{W}_{\omega_{1}} := \{x\in \mathbb{X}_{v}:\,\omega_{\mathrm{nn}}(x)\le \omega_{1}\}$ and  $\mathbb{W}_{\omega_{2}/\omega_{1}} := \{x\in \mathbb{X}_{v}:\, \omega_{1} < \omega_{\mathrm{nn}}(x)\le \omega_{2}\}$.  By conditions \eqref{eq:c2w2_inclusion} and \eqref{eq:dW}, and Lemma \ref{lemma:enlarge_quad},
    $\mathbb{W}_{\omega_{1}}\subseteq \mathbb{E}_{c_{2}}\subseteq \mathcal{X}$ and $\mathbb{W}_{\omega_{2}/\omega_{1}}\subseteq \mathcal{X}\cap \mathbb{X}_{v}$, hence $\mathbb{W}_{\omega_{2}}\subseteq \mathcal{X}\cap \mathbb{X}_{v}$ (safety). Let $x \in \mathbb{W}_{\omega_{1}}$, then using condition \eqref{eq:c2w2_inclusion} and Lemma  \ref{lemma:enlarge_quad}, we have $f(x,w)\in \mathbb{X}_{v}$ for all $w\in W$, and $\omega_{\mathrm{nn}}(f(x,w))\leq w_{2}$ for all $w\in W$, hence $f(x,w)\in \mathbb{W}_{\omega_{2}}$ for all $w\in W$. Also, for $x \in \mathbb{W}_{\omega_{2}/\omega_{1}}$, and  using condition \eqref{eq:dW}, we have $f(x,w) \in \mathbb{X}_{v}$ and $\omega_{\mathrm{nn}}(f(x,w))\leq \omega_{2}-\varepsilon\leq \omega_{2}$ for all $w\in W$, i.e.,  $f(x,w)\in \mathbb{W}_{\omega_{2}}$ for all $w\in W$. Therefore,  $\mathbb{W}_{\omega_{2}}=\mathbb{W}_{\omega_{1}}\bigcup \mathbb{W}_{\omega_{2}/\omega_{1}}$ is invariant under the dynamics of \eqref{eq:System}. Using condition \eqref{eq:c2w2_inclusion} and Lemma \ref{lemma:enlarge_quad}, we have, for  $x\in \mathbb{W}_{\omega_{1}}$ and any arbitrary $\pi \in W^{\mathbb{Z}_{+}}$, $x\in \mathbb{E}_{c_{2}}\Rightarrow \varphi_{x}^{\pi}(k)\in \mathbb{E}_{c_{2}}~ \forall k\in \mathbb{Z}_{+},
    $
    and $\lim_{k\rightarrow \infty}\varphi_{x}^{\pi}(k)=0_{n}$ (i.e., convergence to the  RIS $\mathcal{A}$ for any trajectory starting in $\mathbb{W}_{\omega_{1}}$).  Finally, any trajectory starting in $\mathbb{W}_{\omega_{2}/\omega_{1}}$ enters $\mathbb{W}_{\omega_{1}}$   in finite time.  
This follows by an argument identical to that used in the proof of 
Lemma~\ref{lemma:enlarge_quad}.     
\end{proof}
\subsubsection{Computation of 
\label{sec:CandidateParameters}
candidate values for $c_{2}$, $\omega_{1}$, and $\omega_{2}$}

In the previous sections, we assumed the existence of 
$c_{2}$, $\omega_{1}$, and $\omega_{2}$ 
such that conditions \eqref{eq:dVP}, \eqref{eq:c2w2_inclusion}, and \eqref{eq:dW} hold. 
To obtain a large NN-based ROA estimate, the parameters $c_{2}$ and $\omega_{2}$ should be chosen as large as possible. 
We now describe how suitable \emph{candidate} values for these parameters are computed. 
The estimates are obtained by solving optimization problems derived from conditions~\eqref{eq:dVP}, \eqref{eq:c2w2_inclusion}, and~\eqref{eq:dW}. 
The resulting procedure is conceptually similar to that presented in \cite{lu2024estimating}, while additionally incorporating state constraints.

Define
$$
\bar{c}_{2}=\inf\left\{c\in \intco{c_{1},\infty} \middle\vert \begin{array}{c} \exists (x,w)\in \mathbb{X}_{v}\times W~\text{s.t.}\\
\nu(x)=c~\wedge \\ \bigl (\nu(f(x,w))-\nu(x)> -\varepsilon~ \vee\\ f(x,w)\not\in \mathbb{X}_{v}~\vee~ g(x)\geq 1\bigr)\end{array}  \right\}.
$$
By construction, and assuming the well-definedness of $\bar{c}_{2}$,  for any $ c_{2}\in \intoo{c_{1},\bar{c}_{2}}$, condition \eqref{eq:dVP} holds. 
We therefore select the candidate
$
c_{2}=\bar{c}_{2}-\delta \in \intoo{c_{1},\bar{c}_{2}},
$
for sufficiently small $\delta>0$.

Next define
$$
\overline{\omega}_{1}=\inf\left\{\omega\in \mathbb{R} \middle\vert \begin{array}{c} \exists x\in \mathbb{X}_{v}~\text{s.t.}~ \omega_{\mathrm{nn}}(x)=\omega\wedge \\
 \nu(x)> c_{2} \end{array}\right\},
$$
and
$$
\underline{\omega}_{1}=\sup\left\{\omega\in \intcc{0,\overline{\omega}_{1}}\middle\vert \begin{array}{c} \exists (x,w)\in \mathbb{X}_{v}\times W~\text{s.t.}\\ \omega_{\mathrm{nn}}(x)=\omega~\wedge\\
\omega_{\mathrm{nn}}(f(x,w))-\omega_{\mathrm{nn}}(x) > -  \varepsilon  \end{array}\right\},
$$
where we assume $\underline{\omega}_{1}<\overline{\omega}_{1}$. 
For any $\omega_{1}\in \intoo{\underline{\omega}_{1},\overline{\omega}_{1}}$, 
condition \eqref{eq:c2w2_inclusion} holds except for the constraint involving $\omega_{2}$, and the Lyapunov decrease condition is satisfied on the boundary level set $\omega_{\mathrm{nn}}(x)=\omega_{1}$ within $\mathbb{X}_{v}$. 
To enlarge the certified ROA, $\omega_{1}$ should therefore be chosen as small as possible, while $\omega_{2}$ should be chosen as large as possible so that
$ \omega_{\mathrm{nn}}(x)<\omega_{1}\Rightarrow  \omega_{\mathrm{nn}}(f(x,w))<\omega_{2}$. 
Moreover, by construction, the decrease condition holds at $\omega_{1}$, hence we hope to get a candidate  $\omega_{2}>\omega_{1}$. Accordingly, we select the candidate
$
\omega_{1}=\underline{\omega}_{1}+\delta \in \intoo {\underline{\omega}_{1},\overline{\omega}_{1}}, 
$
for sufficiently small $\delta>0$.

Finally, define
$$
\overline{\omega}_{2}=\inf\left\{\omega\in \intco{{\omega}_{1},\infty} \middle\vert \begin{array}{c} \exists(x,w)\in \mathbb{X}_{v}\times W~\text{s.t.}\\
\omega_{\mathrm{nn}}(x)=\omega~\wedge\\\bigl( \omega_{\mathrm{nn}}(f(x,w))-\omega_{\mathrm{nn}}(x) >- \varepsilon\\
\vee~ f(x,w)\not\in \mathbb{X}_{v}~\vee~ g(x)\geq 1 \bigr)\end{array}  \right\}.
$$

By construction, and assuming well-definedness, for any $\omega_{2}\in\intoo{\omega_{1},\overline{\omega}_{2}
}$, condition \eqref{eq:dW} holds. 
We therefore select the candidate
$
\omega_{2}=\overline{\omega}_{2}-\delta \in \intoo{\omega_{1}, \overline{\omega}_{2}},
$
for sufficiently small $\delta>0$.

The optimization problems associated with the parameters $\bar{c}_{2}$, $\overline{\omega}_{1}$, 
$\underline{\omega}_{1}$, and $\overline{\omega}_{2}$ 
can be solved approximately using a grid-based or sampling-based search over $\mathbb{X}_{v} \times W$. 
The resulting candidate values $c_{2}$, $\omega_{1}$, and $\omega_{2}$ 
are subsequently subjected to formal verification to ensure that conditions 
\eqref{eq:dVP}, \eqref{eq:c2w2_inclusion}, and \eqref{eq:dW} 
hold. 
This verification step can be performed using sound nonlinear verification tools, such as \texttt{dReal}, 
or neural network verification frameworks, including $\alpha,\beta$-CROWN.

\section{Numerical examples}
\label{sec:NumericalExamples}

In this section, we demonstrate the effectiveness of the proposed framework through four numerical examples. 
In all cases, the RIS is given by 
\[
\mathcal{A} = \{0_{n}\}.
\]
For each example, we follow the procedure described in Section~\ref{sec:VFApproximation} to obtain certifiable ROAs. 
All presented examples are non-polynomial systems. Consequently, DOA estimation methods based on SOS polynomial techniques (e.g.,~\cite{xue2020robust}) are not directly applicable. 
This further highlights the generality of the proposed framework.

The neural networks are trained according to the procedure outlined in Section~\ref{sec:PINNLearning}, where the loss-function weights are fixed to $\lambda_{\mathrm{d}}=0.1$ and $\lambda_{\mathrm{pi}}=1$, thereby placing greater emphasis on enforcing the Bellman equation during training. 
All learned neural networks consist of two hidden layers with $\tanh$ activation functions. 
For the 2D examples, each hidden layer contains 20 neurons, whereas for Example~\ref{sec:Pendulum}, each hidden layer contains 30 neurons.  

The initial ROA $\mathbb{E}_{c_1}$ for all examples is ellipsoidal (based on quadratic Lyapunov functions). It is computed using the construction in Appendix~\ref{sec:Ellipsoidal ROAs} for Examples~\ref{sec:TwoMachine} and~\ref{sec:Pendulum}, and via the Lyapunov-based analysis in Appendix~\ref{sec:LyapunovAnalysisPolynomial} for Example~\ref{sec:PolynomiallyStable}.

All data generation, neural-network training, and parameter estimation are carried out in \textsc{Matlab} on a 13th Gen Intel(R) Core(TM) i7-1355U (1.70~GHz) laptop. 
Formal verification is performed using $\alpha,\beta$-CROWN on an NVIDIA H100 NVL GPU (Hopper architecture with 96~GB HBM3 memory). 
The initial ROA $\mathbb{E}_{c_{1}}$, computed via Appendix~\ref{sec:Ellipsoidal ROAs}, is obtained with the aid of the optimization toolbox YALMIP~\cite{lofberg2004yalmip} and the reachability toolbox CORA~\cite{althoff2015introduction}. 
The computational cost of obtaining $\mathbb{E}_{c_{1}}$ is negligible (less than two seconds for all examples).

Finally, the verification of non-certifiable estimated parameters (see Section~\ref{sec:CandidateParameters}) incorporates an additional bisection-based optimization step to further refine the resulting DOA estimates. 
Since the estimated parameters are theoretically optimal, the improvements achieved through this refinement are typically marginal.

\subsection{Comparison Setup}

\subsubsection{Proposed vs. nominal}

As highlighted in the Introduction, Lyapunov functions constructed from nominal dynamics often exhibit a degree of robustness, making them useful for obtaining DOA estimates when disturbances are sufficiently small. 
In Examples~\ref{sec:TwoMachine}, \ref{sec:PolynomiallyStable}, and~\ref{sec:Pendulum}, we compare the proposed NN-based DOA estimates with two alternative approaches:

\begin{enumerate}
\item Optimized ellipsoidal ROAs $\mathbb{E}_{c_2}$ from Section~\ref{sec:EnlargingEllipse}, i.e., the largest ellipsoidal ROAs obtained by enlarging $\mathbb{E}_{c_1}$.
\item ROAs obtained from neural networks trained under nominal dynamics, following the framework of~\cite{serry2025safe}.
\end{enumerate}

The framework in~\cite{serry2025safe} constitutes a special case of the present setting, as it assumes a singleton disturbance set during the learning phase. 
In that case, reachable sets reduce to singleton points, and approximate value-function evaluation requires only one trajectory per data point, i.e., $N_{\mathrm{traj}}=1$. 
The framework in~\cite{serry2025safe} has been shown to provide certifiable large DOA estimates for unperturbed systems.

For each of these examples, we first compute $\mathbb{E}_{c_1}$ and subsequently train neural networks using both the proposed robust framework and the nominal-dynamics framework of~\cite{serry2025safe}. 
Certifiable ROAs are then obtained for both networks following the procedure in Section~\ref{sec:NNVerification}.

For these examples, ROAs are evaluated under two scenarios:
\begin{itemize}
\item \textbf{Scenario 1:} Neglecting uncertainty, i.e., reducing the disturbance set to its center.
\item \textbf{Scenario 2:} Accounting for the full disturbance set.
\end{itemize}

\subsubsection{Proposed vs. parameter-dependent Lyapunov functions}

To further highlight the effectiveness of the proposed framework, we compare the certifiable DOA estimates obtained using our method with (i) the optimized ellipsoidal ROA $\mathbb{E}_{c_{2}}$ and (ii) ROAs derived from parameter-dependent Lyapunov functions~\cite{coutinho2013local,coutinho2010robust}, considering the two-dimensional academic example presented in~\cite{coutinho2013local}.

The framework in~\cite{coutinho2013local} assumes constant parametric uncertainty. 
In contrast, the proposed framework accommodates time-varying uncertainties. 
Therefore, our estimates are, in principle, obtained under more general and potentially more conservative assumptions.

The learning and verification domains, number of sampled points, loss values, and CPU times associated with value-function evaluation, neural-network training, parameter estimation, and verification for all examples are summarized in Table~\ref{tab:NNparameters}.

\begin{table*}[t]
\centering
\caption{Neural network architectures/parameters and corresponding CPU times for the four numerical examples. 
Here, $\mathrm{NN}_{\mathrm{nom}}$ denotes the neural networks trained using the nominal dynamics (i.e., the framework in~\cite{serry2025safe}), whereas $\mathrm{NN}_{\mathrm{prop}}$ denotes those trained using the proposed framework. 
The quantity $t_{\mathrm{ev}}$ represents the CPU time required to evaluate 
$\bigcup_{i=1}^{N_{\mathrm{d}}}\{\tilde{\mathcal{W}}_{N_s}(\{z_{\mathrm{d}}^{(i)}\})\}$, 
$t_{\mathrm{t}}$ is the CPU time required for neural network training, 
$t_{\mathrm{e}}$ is the CPU time required to estimate the parameters $c_{2}$, $\omega_{1}$, and $\omega_{2}$ (for both unperturbed and perturbed scenarios), and 
$t_{\mathrm{v,gpu}}$ is the CPU time required to verify these parameters (for both unperturbed and perturbed scenarios) using a GPU. 
“NA” indicates verification failure. All reported CPU times are given in seconds. }
\label{tab:NNparameters}
\begin{tabular}{|c|c|c|c|c|c|c|c|c|c|c|}
\hline
\multicolumn{2}{|c|}{Parameters/Variables} & $\mathbb{X}_{l}/\mathbb{X}_{v}$ & $N_{s}$ & $N_{\mathrm{d}}$/$N_{\mathrm{pi}}$ & $N_{\mathrm{traj}}$ & $\mathrm{Loss}(\theta)$& $t_{\mathrm{cpu,ev}}$& $t_{\mathrm{t}}$& $t_{\mathrm{e}}$& $t_{\mathrm{v,gpu}}$\\
\hline
\multirow{2}{*}{Ex.  \ref{sec:TwoMachine}}
& $\mathrm{NN}_{\mathrm{nom}}$ & \multirow{2}{*}{0.7\Hintcc{-1_{2},1_{2}}} & \multirow{2}{*}{500} & \multirow{2}{*}{5000} &1&9e-6&1.9&227.0&0.6/6.2&16.0/NA \\ \cline{2-2}\cline{6-6}\cline{7-7}\cline{8-8}\cline{9-9}\cline{10-10}\cline{11-11}
& $\mathrm{NN}_{\mathrm{prop}}$ &  &  &  &1000&4.0e-6&278.0&218.7&0.6/6.1&17.2/19.5\\
\hline\multirow{2}{*}{Ex. \ref{sec:PolynomiallyStable}}
& $\mathrm{NN}_{\mathrm{nom}}$ & \multirow{2}{*}{\Hintcc{-1_{2},1_{2}}} & \multirow{2}{*}{500} & \multirow{2}{*}{5000} &1&4.3e-5&0.5&215.0&0.6/6.5&22.9/22.9\\ \cline{2-2}  \cline{6-6}\cline{7-7}\cline{8-8}\cline{9-9}\cline{10-10}\cline{11-11}
& $\mathrm{NN}_{\mathrm{prop}}$ &  &  &  &1000&7.8e-5&245.4&211.0&0.6/7.8&15.7/38.6\\
\hline
\multirow{2}{*}{Ex. \ref{sec:Pendulum}}
& $\mathrm{NN}_{\mathrm{nom}}$ & \multirow{2}{*}{$\intcc{-\frac{\pi}{4},\frac{\pi}{4}}\times\Hintcc{-1_{2},1_{2}}$} & \multirow{2}{*}{500} &\multirow{2}{*}{7000}&1&1.4e-4&0.9&656.8&66.5/399.0&26.3/NA \\ \cline{2-2}\cline{6-6}\cline{7-7}\cline{8-8}\cline{9-9}\cline{10-10}\cline{11-11}
& $\mathrm{NN}_{\mathrm{prop}}$ &  &  &  &3000&1.0e-4&683.7&622.2&65.5/385.5&23.4/43.0\\
\hline
\multirow{2}{*}{Ex. \ref{sec:Rational}}
& \multirow{2}{*}{$\mathrm{NN}_{\mathrm{prop}}$} & \multirow{2}{*}{$\intcc{-3,3}\times\intcc{-3,3}$} & \multirow{2}{*}{1000} &\multirow{2}{*}{5000/30000}&\multirow{2}{*}{1000}&\multirow{2}{*}{2.0e-6}&\multirow{2}{*}{1369.0}&\multirow{2}{*}{613.3}&\multirow{2}{*}{7.4}&\multirow{2}{*}{42.3}  \\
&  &  &  &  & & & & & & \\
\hline
\end{tabular}
\end{table*}

\subsection{A perturbed two-machine power system}
\label{sec:TwoMachine}
We consider a discrete and perturbed version of the two-dimensional two-machine power system studied in \cite{vannelli1985maximal,willems1968improved}. The discrete version is obtained through Euler discretization of its continuous-time version, which yields
$$
f(x,w)=    \begin{pmatrix}
    x_{1}+\Delta_{t}  x_{2}\\
       x_{2}-\Delta_{t}( w{x_{k}}+\sin(x_{1}+\pi/3)-\sin(\pi/3))
    \end{pmatrix},
$$
where $w\in W=[0.25,0.75]$, and  $\Delta_{t} = 0.2$.  he safe set is the intersection of two safe sets: 
 $$
 \mathcal{X}_{1}=\mathbb{R}^2\setminus  (\begin{bmatrix}
     -0.5\\0.5
 \end{bmatrix}+\frac{1}{4}\mathbb{B}_{2}),~
\mathcal{X}_{2}=\mathbb{R}^2\setminus  (\begin{bmatrix}
    0\\-0.5
\end{bmatrix}+\frac{1}{4}\mathbb{B}_{2}).
$$
These sets are 1-sublevel sets of the functions
 \begin{align*}
g_{\mathcal{X}_{1}}(x)&=1+\frac{1}{4^2}-((x_{1}+0.5)^2+(x_{2}-0.5)^2),\\
g_{\mathcal{X}_{2}}(x)&=1+\frac{1}{4^2}-(x_{1}^2+(x_{2}+0.5)^2),
\end{align*}
respectively.
 From the procedure in Appendix \ref{sec:CandidateParameters}, we obtained an ROA $\mathbb{E}_{c_1}$, 
For this example, we choose $\alpha(x)=0.1 \nu(x)$\footnote{If the $\ell_{p}$-stability condition holds for some $0<p\le 1$ (for instance, in the case where $\mathcal{A}$ is exponentially stable), and the function $\nu$ defining the ROAs $\mathbb{E}_{c_{1}}$ and $\mathbb{E}_{c_{2}}$ is a Lyapunov function, then it can be beneficial to select $\alpha(x) = c\,\nu(x)$, where $c > 0$ is a scaling parameter. 
With this choice, the sublevel sets of the value function $\mathcal{W}$, and consequently those of its finite-sample approximation $\tilde{\mathcal{W}}_{N_{s}}$ and neural network approximation $\omega_{\mathrm{nn}}$, tend to reflect the geometric structure of the sublevel sets of $\nu$, particularly in a neighborhood of the RIS $\mathcal{A}$. 
This alignment may improve the quality of the NN-based ROA estimates.
}. For this example, the values $F(\{x\})$, for $x \in \mathbb{R}^{2}$, are hyper rectangles. 
Accordingly, the embedding $\mathcal{T}$ employed is that described in Remark~\ref{rem:transformation} for hyper-interval images.

\subsection{An uncertain nonlinear system with local polynomial stability}
\label{sec:PolynomiallyStable}
Consider the uncertain discrete-time system
\begin{equation}\label{eq:cos_poly_sys}
x_{k+1}=\bigl(\cos(\|x_k\|^{2})-w\|x_k\|^{2}\bigr)\,x_k,
\qquad x_k\in\mathbb{R}^{2},
\end{equation}
where
$
w\in\intcc{w_{\min},w_{\max}}
$
with $w_{\min}=1$ and $w_{\max}=2$, $\mathcal{X}=\{x\in \mathbb{R}^{2}| x_{1}+y_{1}<1\}$. It can be shown that for $\nu(\cdot)=\norm{\cdot}^{2}$ and $c_{1}=\min\{{1}/{w_{\min}}-\varepsilon,-w_{\max}+\sqrt{w_{\max}^2+2}\}$ for some small $\varepsilon>0$, the region $\mathbb{E}_{c_{1}}$ is a region of attraction and that the convergence rate is of order $1/2$ (see  Appendix \ref{sec:LyapunovAnalysisPolynomial}). Therefore, the system is locally uniformly $\ell_{p}$ stable with $p>2$. We define $\alpha(x)=\norm{x}^4$.  In this example, the values $F(\{x\})$, for $x \in \mathbb{R}^{2}$, are line segments of the form
$\{\, u + t v \mid t \in \intcc{-1,1} \,\}$, where $u \in \mathbb{R}^{2}$ denotes the center and $v \in \mathbb{R}^{2}$ specifies the direction vector. 
Accordingly, we employ the embedding $\mathcal{T}$ defined as, for $X = \{\, u + t v \mid t \in \intcc{-1,1} \,\}$,  
\[
\mathcal{T}(X) =
\begin{bmatrix}
u \\
v
\end{bmatrix}.
\]

\subsection{A rigid rod with gravity and saturated torque}
\label{sec:Pendulum}

Consider a discrete-time rigid rod system with mass $m=1$ and length $l=1$. The mass moment of inertia about the geometric center is  $I=ml^{2}/12$. The uncertain variable $w$ represents the distance between the loosely pinned end of the rod and its geometric center.  The system dynamics are given by
\begin{align}
\tilde{f}(x,u,w)=
\begin{pmatrix}
x_{1} + \Delta_{t} x_{2} \\
x_{2} + \Delta_{t} \dfrac{1}{I+mw^2} \left( x_{3} - mgw \sin(x_{1}) \right) \\
x_{3} + \Delta_{t} u
\end{pmatrix},
\end{align}
where $x_{1}$ denotes the angular position, $x_{2}$ the angular velocity, and $x_{3}$ the applied torque. The control input $u$ represents the torque rate. The sampling time is $\Delta_{t}=0.2$, and $g=9.81$ denotes the gravitational acceleration.  The uncertain parameter $w$ belongs to the interval $W = \intcc{0.5l - 0.04l,\; 0.5l + 0.04l}$.
The torque rate is defined via a state-feedback control law $\tilde{u}(x) = Kx$, where the feedback matrix $K$ is chosen such that the origin of the closed-loop unperturbed system is locally exponentially stable\footnote{The matrix $K$ can be computed by solving the discrete-time Riccati equation associated with the linearized dynamics about the origin.}. The resulting closed-loop dynamics are
$f(x,w) = \tilde{f}(x,\tilde{u}(x),w)$.
We impose the following state and input constraints:
$|x_{1}| < \frac{\pi}{4}$, 
$|x_{2}| < 1$,
$|x_{3}| < 1$,
and 
$|\tilde{u}(x)| < 5$. These constraints can be expressed as intersections of four safe sets $\mathcal{X}_{i}$ defined as strict $1$-sublevel sets of the functions $g_{1}(x)={x_{1}^{2}}/{(\pi/4)^{2}}$,  $g_{2}(x)=x_{2}^{2}$,  $g_{3}(x)={x_{3}^{2}}$, 
and
$g_{4}(x)={x^{\intercal}K^{\intercal}Kx}/{5^{2}}$. Following the procedure described in Appendix~\ref{sec:Ellipsoidal ROAs}, we compute an initial ellipsoidal ROA $\mathbb{E}_{c_{1}}$ characterized by the function
$\nu(x)=x^{\intercal}Px$. For this example, we select $\alpha(x)=0.1\,\nu(x)$. Finally, the values $F(\{x\})$, for $x \in \mathbb{R}^{3}$, are hyper-rectangles, so we use the embedding in described in Remark~\ref{rem:transformation} for hyper-interval images. 
\subsection{A perturbed rational system}
\label{sec:Rational}
We consider the rational uncertain system studied in \cite{coutinho2013local}:
$$
f(x,w)=    \begin{pmatrix}
    x_{1}- \Delta_{t}\frac{(x_{1}+x_{2}^3)}{1+x_{2}^2}\\
    x_{2}+\Delta_{t}\frac{x_{1}^3-(0.25+w)x_{2}}{1+x_{2}^2}
    \end{pmatrix},
$$
where $\Delta_{t}=0.01$, $w\in \intcc{-0.15,0.15}$ and $\mathcal{X}=\mathbb{R}^{2}$. From the procedure in Appendix \ref{sec:CandidateParameters}, we obtained an ROA $\mathbb{E}_{c_1}$,  and selected $\alpha(x)=\Delta_{t} \nu(x)$. The values $F(\{x\})$, for $x \in \mathbb{R}^{3}$, are hyper-rectangles, so we use the hyper-interval-based embedding in Remark~\ref{rem:transformation}.

\subsection{Results and discussion}

Figures~\ref{fig:twomachine}, \ref{fig:polynomial}, and~\ref{fig:pendulum} illustrate the ROA estimates under the two considered scenarios. 
When uncertainty is neglected (Scenario~1), the nominal-dynamics framework of~\cite{serry2025safe} yields larger certifiable ROAs compared to both the ellipsoidal ROAs and the proposed framework. 
This behavior is expected, as the proposed method explicitly accounts for uncertainty during training, leading to inherently more conservative estimates when disturbances are ignored. 
In this scenario, the proposed approach performs comparably to, and in some cases slightly better than, the optimized ellipsoidal method.

When disturbances are fully accounted for (Scenario~2), the proposed framework consistently provides the largest certifiable DOA estimates. 
In contrast, the nominally trained neural networks fail to produce certifiable estimates for Examples~\ref{sec:TwoMachine} and~\ref{sec:Pendulum}. 
Notably, the optimized ellipsoidal ROAs still provide reasonable estimates across all three examples, demonstrating their inherent robustness despite their structural limitations.

From a computational perspective, Table~\ref{tab:NNparameters} shows that the proposed framework is comparable to the nominal-based approach in terms of neural-network training, parameter estimation, and verification times. 
The primary computational overhead arises from evaluating the (approximate) ground-truth value function $\mathcal{W}$ at the collocation points, i.e., $\bigcup_{i=1}^{N_{\mathrm{d}}}
\left\{
\tilde{\mathcal{W}}_{N_s}(\{z_{\mathrm{d}}^{(i)}\})
\right\}$. As discussed in Section~\ref{sec:ReachSetApproximation}, these evaluations require reachable-set computations. 
In the present examples, reachable sets are approximated using randomly generated trajectories, whose number must be sufficiently large to ensure accuracy. 
This results in the typical accuracy--computational cost tradeoff inherent to numerical approximation schemes.

Figure~\ref{fig:rational} further demonstrates, for Example~\ref{sec:Rational}, that the proposed approach outperforms both the optimized ellipsoidal estimate $\mathbb{E}_{c_{2}}$ and the ROA derived in~\cite{coutinho2013local}. 
Importantly, this improvement is achieved despite the proposed framework accommodating time-varying disturbances, whereas~\cite{coutinho2013local} considers only constant parametric uncertainty. 
This comparison underscores the effectiveness and flexibility of the proposed method.

Finally, Table~\ref{tab:NNparameters} indicates that data generation requires more computational time for Example~\ref{sec:Rational} compared to the other examples. 
This is attributed to the use of a larger time horizon $N_{s}$, as the system is obtained via Euler discretization with a finer time step.

\begin{figure}
    \centering
    \includegraphics[width=0.8\linewidth]{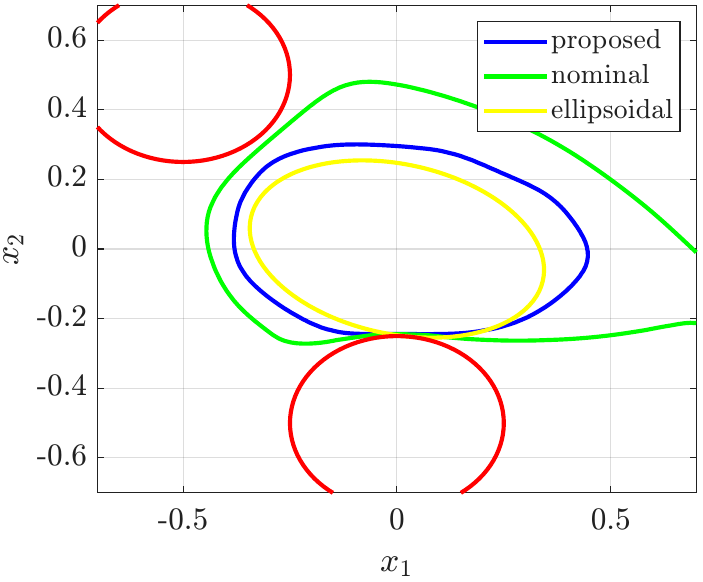}
    \includegraphics[width=0.8\linewidth]{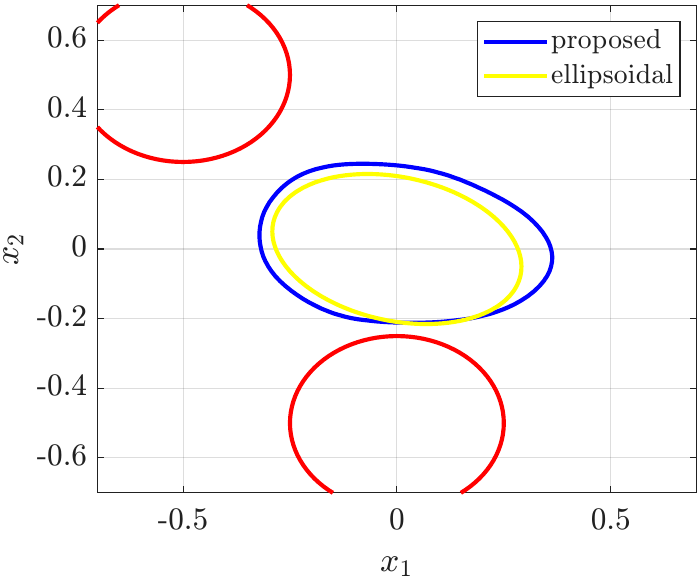}
    \caption{Certified ROAs  for Example~\ref{sec:TwoMachine} under Scenario~1 (uncertainty neglected, top) and Scenario~2 (full disturbance considered, bottom). 
Red circles indicate the state constraints. 
The nominally trained neural network fails to yield a certifiable ROA under Scenario~2.}\label{fig:twomachine}
\end{figure}
\begin{figure}
    \centering
    \includegraphics[width=0.8\linewidth]{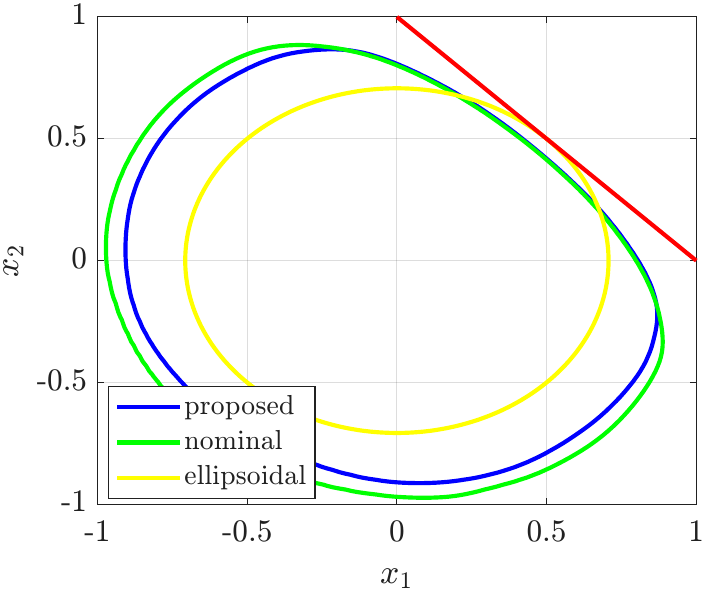}
    \includegraphics[width=0.8\linewidth]{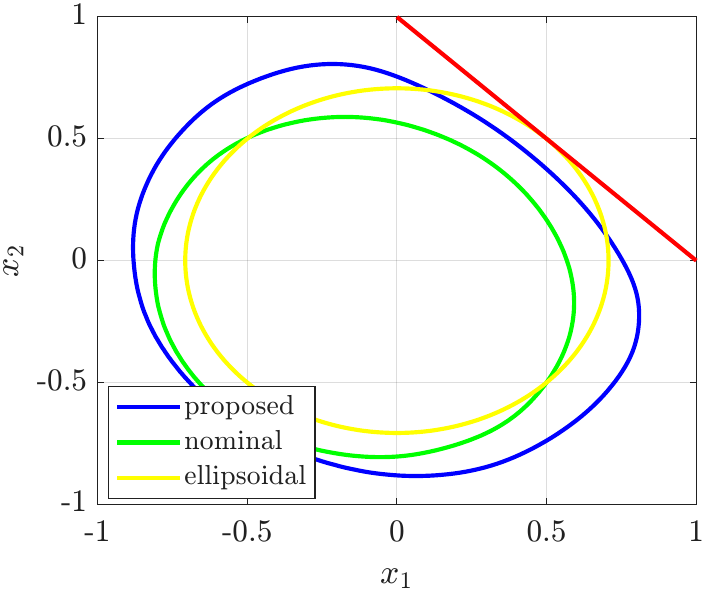}
    \caption{Certified ROAs  for Example~\ref{sec:PolynomiallyStable} under Scenario~1 (uncertainty neglected, top) and Scenario~2 (full disturbance considered, bottom). 
Red lines indicate the state constraints.}
    \label{fig:polynomial}
\end{figure}

\begin{figure}
    \centering
    \includegraphics[width=0.8\linewidth]{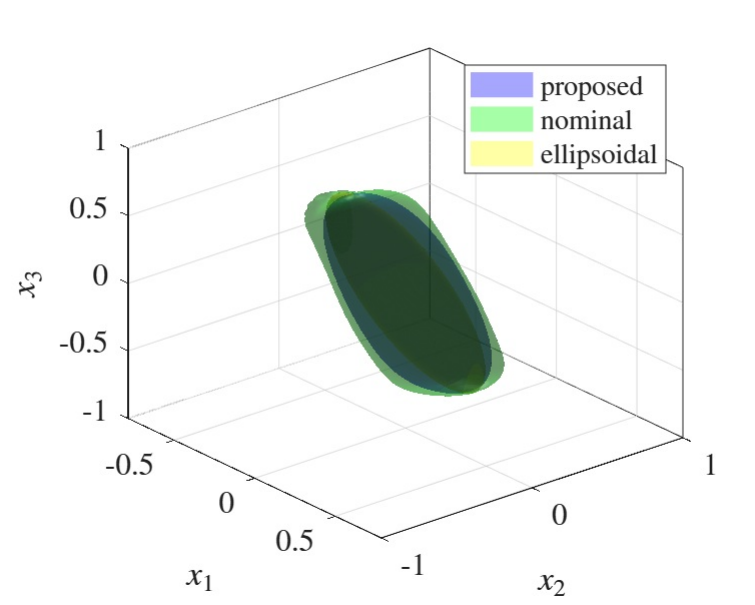}
    \includegraphics[width=0.8\linewidth]{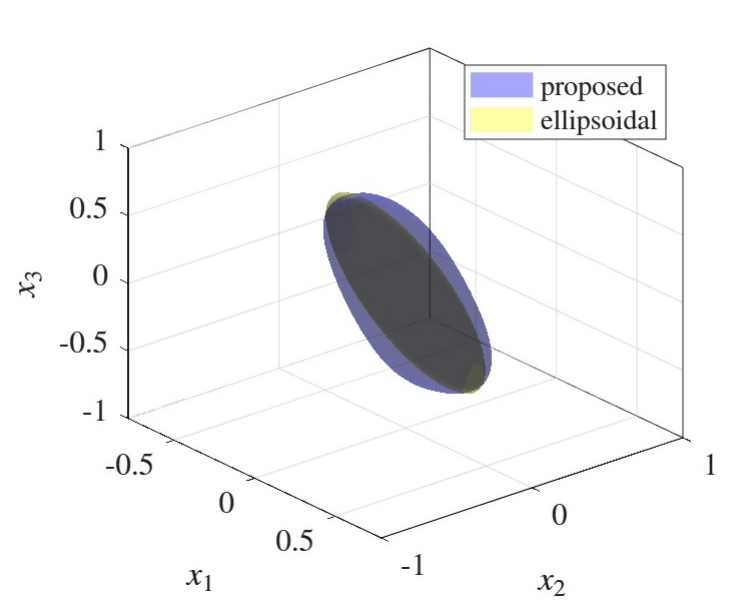}
    \caption{Certified ROAs  for Example~\ref{sec:Pendulum} under Scenario~1 (uncertainty neglected, top) and Scenario~2 (full disturbance considered, bottom).  
The nominally trained neural network fails to yield a certifiable ROA under Scenario~2.}
    \label{fig:pendulum}
\end{figure}

\begin{figure}
    \centering
    
    \includegraphics[width=0.7\linewidth]{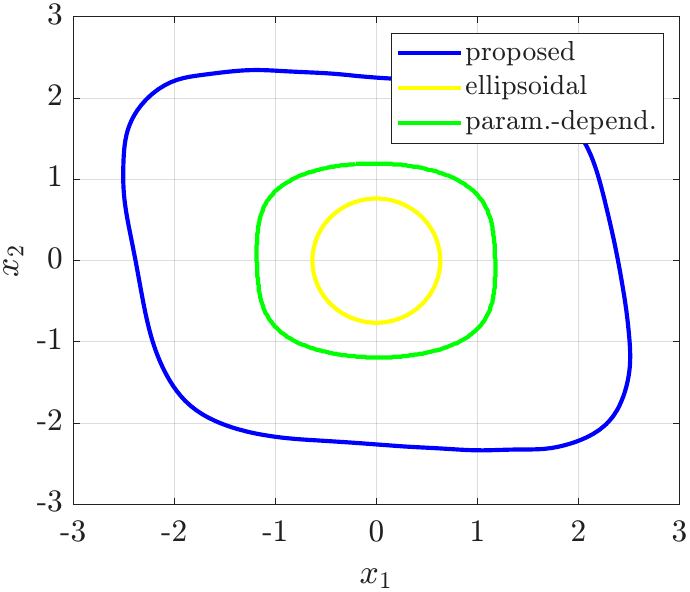}
    \caption{Certified ROAs  for Example \ref{sec:Rational}. The ROA based on parameter-dependent Lyapunov functions is digitized from \cite{coutinho2013local}.}
    \label{fig:rational}
\end{figure}

\section{Conclusion}
\label{sec:Conclusion}
In this paper, we proposed a novel framework for the accurate estimation of safe and robust DOAs for discrete-time nonlinear uncertain systems with continuous dynamics, open safe sets, compact disturbance sets, and uniformly locally $\ell_p$-stable RISs. The DOAs were characterized via newly introduced value functions defined on metric spaces of compact sets, shown to be the unique solutions of corresponding Bellman-type equations. Building on this characterization, we developed a physics-informed NN framework that learns the associated value functions by embedding the derived Bellman-type equations directly into the training process. To obtain certifiable estimates of the safe robust DOAs from the learned neural approximations, we introduced a verification procedure that leverages existing formal verification tools. The effectiveness and applicability of the proposed methodology were demonstrated through numerical examples.

Future work will explore several research directions motivated by the present framework, including the study of safe robust DOAs for continuous-time systems, the characterization and estimation of null controllability regions for controlled and perturbed systems, and the development of more efficient NN verification algorithms for accurate and scalable ROA estimation.

\bibliographystyle{ieeetr}

\appendices
 \section{Proof of Lemma \ref{lem:enlargement_of_compact}}\label{Proof_lem:enlargement_of_compact}
 \begin{proof}
   As $X$ is open, then for every $x\in \Omega$, there exists $\delta_{x}>0$ such that $x+\delta_{x}\mathbb{B}_{n}\subseteq X$. Therefore,
   $
\Omega\subseteq\bigcup_{x\in \Omega} x+({1}/{2})\delta_{x}\mathbb{B}_{n}\subseteq  \bigcup_{x\in \Omega} x+\delta_{x}\mathbb{B}_{n}\subseteq X.
   $
   Using the compactness of $\Omega$, there exists a finite $N\in \mathbb{N}$ and a finite set $\{x_{i}\}_{i=1}^{N}\subseteq \Omega$ such that 
   $
\Omega \subseteq \bigcup_{i\in \intcc{1;N}} x_{i}+({1}/{2})\delta_{x_{i}}\mathbb{B}_{n}.
   $
  Let 
  $
\Omega_{i}\defas\Omega \bigcap (x_{i}+({1}/{2})\delta_{x_{i}}\mathbb{B}_{n}),~i\in \intcc{1;N}
  $. 
Note that $\bigcup_{i\in \intcc{1;N}}\Omega_{i}=\Omega$ as $\Omega \subseteq \bigcup_{i\in \intcc{1;N}} x_{i}+({1}/{2})\delta_{x_{i}}\mathbb{B}_{n}$.   Define
$
\delta\defas\min_{i\in \intcc{1;N}}({1}/{2})\delta_{x_{i}}>0.
 $
  Then,
  $
\Omega_{i}+\delta\mathbb{B}_{n}\subseteq x_{i}+(\delta +({1}/{2})\delta_{x_{i}})\mathbb{B}_{n}\subseteq x_{i}+\delta_{x_{i}}\mathbb{B}_{n}\subseteq X,~i\in \intcc{1;N}. 
  $
Consequently
  $
\Omega+\delta \mathbb{B}_{n}=\bigcup_{i\in \intcc{1;N}}\Omega_{i}+\delta \mathbb{B}_{n}=\bigcup_{i\in \intcc{1;N}}\left(\Omega_{i}+\delta \mathbb{B}_{n}\right)\subseteq \bigcup_{i\in \intcc{1;N}}x_{i}+\delta_{x_{i}}\mathbb{B}_{n}\subseteq X.
  $
  
 \end{proof}
 \section{Proof of Lemma \ref{lem:CompactSetsMetric}} \label{Proof_lem:CompactSetsMetric}
 \begin{proof}
    $(\mathcal{K}(X), \dh)$ being a metric space is an elementary result, for example see \cite[Chapter 3.2]{beer1993topologies}.
    
    When $X$ is open, let $\Omega\in \mathcal{K}(X)$ by Lemma~\ref{lem:enlargement_of_compact} there is a $\delta>0$ with $\Omega+\delta\mathbb B_n \subseteq X$, by definition it follows that any $\Gamma\in \mathcal{K}(\mathbb R^n)$ with $\dh(\Omega, \Gamma) <{\delta}/{2}$ must have $\Gamma \subseteq \Omega+\delta\mathbb B_n \subseteq X$.
    This shows that $\mathcal{K}(X) $ is open.
\end{proof}
\section{Proof of Lemma \ref{lem:FCts}}
\label{Proof_lem:FCts}
\begin{proof}
    We will show that $F$ is continuous at some arbitrary $X\in \Comp$. Fix $r\in \intoo{0,\infty}$. We have $X+r\mathbb{B}_{n}$ is compact.
    Let $\varepsilon>0$ be arbitrary and pick $\delta\in \intoo{0,r}$ such that for all $x_{1},x_{2}\in X+r\mathbb{B}_{n}$ and $w\in W$:
        $
\norm{x_{1}-x_{2}}\leq \delta \Rightarrow \norm{f(x_{1},w)-f(x_{2},w)}\leq \varepsilon.
    $
Such $\delta$ exists due to the uniform continuity of $f$  over $(X+r\mathbb{B}_{n})\times W$. Let $Y\in \mathcal{K}(\mathbb{R}^{n})$ be such that $\dh(X,Y)\leq \delta$. Note that $X,Y\subseteq X+r \mathbb{B}_{n}$. Let $z_{1}\in F(X)$, then there exists $x_{1}\in X$ and $w\in W$ such that $z_{1}=f(x_{1},w)$. As $\dh(X,Y)\leq \delta$, then there exists $x_{2}\in Y$ such that $\norm{x_{1}-x_{2}}\leq \delta$. Define $z_{2}=f(x_{2},w)\in f(Y,W)=F(Y)$. Then, $\norm{z_{1}-z_{2}}=\norm{f(x_{1},w)-f(x_{2},w)}\leq \varepsilon$. Similarly, and using the symmetry of the Hausdorff distance,  for any $z_{2}\in F(Y)$, there exists $z_{1}\in F(X)$ such that $\norm{z_{1}-z_{2}}\leq \varepsilon$. Therefore, $\dh(F(X),F(Y))<\varepsilon$ and that completes the proof.
\end{proof}
\section{Proof of Lemma  \ref{lem:SemiGroup}}
\label{Proof_lem:SemiGroup}
\begin{proof}
    By induction we see that $F(X)=\bigcup_{(x,w)\in X\times W}\{f(x,w)\}=\bigcup_{(x,\pi)\in X\times W^{\mathbb{Z}_{+}}}\{\varphi_{x}^{\pi}(1)\}= \mathcal{R}(X,1)$. 
    Assume $\mathcal{R}(X,k)=\mathcal{R}(F(X),k-1)=F^{k}(X)$.   $F^{(k+1)}(X)=F\circ F^k(X)= F(\mathcal{R}(X,k)) = \bigcup_{(x,\pi,w) \in X\times W^{\mathbb{Z_+}}\times W}\{f(\varphi_{x}^{\pi}(k),w)  \}$.
    Since for every $\pi \in W^{\mathbb{Z_+}}$ there is $w\in W$ with $\pi(k+1)=w$ we get $F^{(k+1)}(X)=\mathcal{R}(X,k+1)$. 
    Moreover, we see that
    $F^{(k+1)}(X)=F^k\circ F(X)= \mathcal{R}(F(X),k) =\bigcup_{(x,w,\pi) \in X\times W\times  W^{\mathbb{Z_+}}} \{\varphi_{f(x,w)}^{\pi}(k) \}$.
    Since for every $\pi \in W^{\mathbb{Z_+}}$ there is $\mu \in W^{\mathbb{Z_+}}$ and a $w\in W$ with $\mu (1)=w$ and $\pi(j-1)=\mu (j)$ for $j\in \intcc{2;k}$, we get $F^{(k+1)}(X)=\mathcal{R}(F(X),k)$. 
    
\end{proof}
\section{Proof of Lemma \ref{Lem:SupIsCts}}
\label{Proof_Lem:SupIsCts}
\begin{proof}
  Fix $X\in D$ and let $r>0$ be an arbitrary positive number. Then $g$ is uniformity continuous over $X+r\mathbb{B}_{n}$. Let $\varepsilon>0$ be arbitrary, then there exists $\delta>0$ with $\delta<r$ such that for all $x,y\in X+r\mathbb{B}_{n}$ satisfying $\norm{x-y}\leq \delta$, we have $\abs{g(x)-g(y)}\leq \varepsilon$. Let  $Y\in D$ satisfy $\dh(X,Y)\leq \delta$, then $X,Y\subseteq X+r\mathbb{B}_{n}$. Let $x\in X$ be such that $g(x)=G(X)$, which exists due to the compactness of $X$. As $\dh(X,Y)\leq \delta$, there exists $y\in Y$ such that $\abs{g(x)-g(y)}\leq \varepsilon$. Therefore, we have 
  $G(X)=
 g(x)\leq g(y)+\varepsilon\leq G(Y)+\varepsilon\Rightarrow G(X)-G(Y)\leq \varepsilon.
  $
Similarly, let $y\in Y$ be such that $g(y)=G(Y)$. As $\dh(X,Y)\leq \delta$, there exists $x\in X$ such that $\norm{x-y}\leq \delta$. Therefore, $\abs{g(y)-g(x)}\leq \varepsilon$, implying
$
G(Y)=g(y)\leq g(x)+\varepsilon\leq G(X)+\varepsilon\Rightarrow G(Y)-G(X)\leq \varepsilon$.
Consequently,
$
\abs{G(X)-G(Y)}\leq \varepsilon
$
and that completes the proof.
\end{proof}

\section{Ellipsoidal ROAs for 
exponentially stable equilibirum points}
\label{sec:Ellipsoidal ROAs}
In what follows, we restrict our attention to singleton invariant sets, i.e., equilibrium points, that are uniformly locally exponentially stable. 
Without loss of generality, we assume that the equilibrium point is located at the origin, i.e., $\mathcal{A} = \{0_n\}$.

We find the initial ROA  using common quadratic Lyapunov functions (see, e.g., \cite{amato2002note}) and linearization about the origin. Herein, $\norm{\cdot}$ denotes the Euclidean norm,  and we assume $f$ to be twice-continuously differentiable. 

Given $x\in \mathbb{R}^{n}$  and $A\in \mathbb{R}^{n\times m}$, $\abs{x}\in \mathbb{R}^{n}_{+}$ and $\abs{A}\in \mathbb{R}_{+}^{n\times m}$ are defined as $|x|_{i}\defas |x_{i}|,~i\in \intcc{1;n}$, and $\abs{A}_{i,j}\defas \abs{A_{i,j}},~(i,j)\in \intcc{1;n}\times \intcc{1;m}$, respectively. 
Let $\mathcal{S}^{n}$ denote the set of $n\times n$ real symmetric matrices. 
Given $A\in \mathcal{S}^{n}$, $\underline{\lambda}(A)$ and $\overline{\lambda}(A)$ denote the minimum and maximum eigenvalues of $A$, respectively.  
Let $\mathcal{S}_{++}^{n}$ denote the set of $n\times n$ real symmetric positive definite matrices, that is $\mathcal{S}_{++}^{n}=\Set{A\in \mathcal{S}^{n}}{\underline{\lambda}(A)>0}$. Given $A\in \mathcal{S}_{++}^{n}$, $A^{\frac{1}{2}}$ denotes the unique $K\in \mathcal{S}_{++}^{n}$ satisfying $A=K^2$ \cite[p.~220]{Abadir2005matrix}.

The following result illustrates the construction of ellipsoidal ROAs.
\begin{theorem} \label{thm:InitialEllipsoidalROA}
 Consider a hyper-rectangular domain $$\mathcal{B}=\Hintcc{-R_{\mathcal{B}}, R_{\mathcal{B}}}\subseteq  \mathcal{X}\cap \mathbb{X}_{v},
 $$
 containing the origin in its interior\footnote{ Such a hyper-rectangle exists due to the openness of $\mathcal{X}$.}. Define 
 $$
h(x,w)\defas f(x,w)-D_{x}f(0_{n},w)x,~(x,w)\in \mathbb{R}^{n}\times \mathbb{R}^{m}, 
$$
and let  $\eta_{\mathcal{B}}\in \mathbb{R}_{+}^{n}$be such that\footnote{$\eta_{\mathcal{B}}$ exits due to the twice-continuous differentiability of $f$ and   can be estimated using interval arithmetic.} 
$$
\abs{h(x,w)}\leq \frac{\norm{x}^{2}}{2}\eta_{\mathcal{B}} ~\forall (x,w)\in \mathcal{B} \times W.
$$
Let $\mathbf{A}\subseteq \mathbb{R}^{n\times n}$ be a polytopic matrix set\footnote{This can be obtained using  interval arithmetic.} such that
$$
D_{x}f(0,w)\in \mathbf{A}~\forall w\in W, 
$$
and let  $A_{1},\ldots A_{q}\in \mathbb{R}^{n\times n}$ be the vertices of $\mathbf{A}$, so that 
$
\mathbf{A}=\mathrm{conv}\{A_{1},\ldots A_{q}\}.
$
Assume  there exists  $P\in \mathcal{S}_{++}^{n}$ such that  
$$
Q_{i}\defas -(A^{\intercal}_{i}PA_{i}-P)\in \mathcal{S}_{++}^{n},~i\in \intcc{1;q}.
 $$
 Define the  candidate Lyapunov function 
 $$
{\nu}(x)\defas x^{\intercal}Px,~x\in \mathbb{R}^{n}.
$$
Moreover, define the parameters 
$$
c_{1}\defas \left(\frac{-e_{2}+\sqrt{e_{2}^2+4e_{1} d}}{2e_{1}}\right)^{2},
$$
$$e_{1}\defas\frac{\norm{\abs{P^{\frac{1}{2}}}\eta_{\mathcal{B}}}^{2}}{4\underline{\lambda}(P)},
$$
$$
d=\min_{i\in \intcc{1;q}}\underline{\lambda}(Q_{i})-\epsilon,
$$
for some sufficiently small $\epsilon>0$ that satisfies $\epsilon<\min_{i\in \intcc{1;q}}\underline{\lambda}(Q_{i})$, and 
$
\epsilon<\overline{\lambda}(P),
$
and  $e_{2}\geq 0$\footnote{The parameter $e_{2}$ can be estimated by means of interval arithmetic.} satisfies  
$$
e_{2}\geq \norm{\abs{P^{\frac{1}{2}}}\eta_{\mathcal{B}}}\norm{P^{\frac{1}{2}}(D_{x}f(0_{n},w))P^{-\frac{1}{2}}}~\forall w\in W.
$$
Additionally, define
$$
k_{2}\defas \min_{i\in \intcc{1;n}} \frac{ R_{\mathcal{B},i}^{2}}{P^{-1}_{i,i}},
$$
and    
$$
c_{1}\defas \min\{k_{1},k_{2}\}.
$$
Then, the ellipsoidal set 
$
\mathbb{E}_{c_{1}}\defas \{x\in \mathbb{R}^{n}|\nu(x)\leq c_{1}\}
$ is a safe robust ROA within $\mathcal{B}\subseteq \mathcal{X}\cap \mathbb{X}_{v}$.
 \end{theorem}
\begin{proof}
For $(x,w)\in \mathbb{R}^{n}\times W$,
$$
f(x,w)= D_{x}f(0_{n},w)x+h(x,w).
$$
Therefore, for $(x,w)\in \mathbb{R}^{n}\times W$ (the arguments of $h$, $f$, and $D_{x}f$ are dropped for convenience), 
\begin{align*}
\nu(f)-\nu(x) =&(f)^{\intercal}P(f)-x^{\intercal}Px\\
=&((D_{x}f)x+h)^{\intercal}P((D_{x}f)x+h)-x^{\intercal}Px\\
=& x^{\intercal}((D_{x}f)^{\intercal}P(D_{x}f)-P)x\\
&+ x^{\intercal}(D_{x}f)^{\intercal}Ph+h^{\intercal}P(D_{x}f)x+h^{\intercal}Ph.
\end{align*}
Note that the term 
$
x^{\intercal}((D_{x}f)^{\intercal}P(D_{x}f)-P)x
$
is a convex function of $D_{x}f$. Hence, for all $w\in W$,
\begin{align*}
x^{\intercal}((D_{x}f)^{\intercal}P(D_{x}f)-P)x\leq& \sup_{A\in  \mathbf{A}} x^{\intercal}(A^{\intercal}PA-P)x\\
&=\sup_{i\in \intcc{1;q}} x^{\intercal}(A_{i}^{\intercal}PA_{i}-P)x\\
&= -\sup_{i\in \intcc{1;q}} x^{\intercal}Q_{i}x.
\end{align*}
Therefore,
$$
\nu(f)-\nu(x)\leq  -\min_{i\in \intcc{1;q}}\underline{\lambda}(Q_{i})\norm{x}^{2}+2x^{\intercal}(D_{x}f)^{\intercal}Ph+h^{\intercal}Ph.
$$ 
For $(x,w)\in \mathcal{B}\times W$, 
\begin{align*}
2x^{\intercal}(D_{x}f)^{\intercal}Ph=&2(P^{-1/2}P^{\frac{1}{2}}x)^{\intercal}(D_{x}f)^{\intercal}P^{\frac{1}{2}}P^{\frac{1}{2}}h\\
&=(P^{\frac{1}{2}}x)^{\intercal}P^{-1/2} (D_{x}f)^{\intercal}P^{\frac{1}{2}}(2P^{\frac{1}{2}}h)\\
&\leq \norm{P^{\frac{1}{2}}x}\norm{P^{\frac{1}{2}}(D_{x}f)P^{-1/2}}\norm{\abs{P^{\frac{1}{2}}}\eta_{\mathcal{B}}}\norm{x}^{2}\\
& \leq \sqrt{\nu(x)}\norm{\abs{P^{\frac{1}{2}}}\eta_{\mathcal{B}}}\norm{P^{\frac{1}{2}}(D_{x}f)P^{-\frac{1}{2}}}\norm{x}^{2}.
\end{align*}
 Also, 
 \begin{align*}
 h^{\intercal}Ph =& \norm{P^{\frac{1}{2}}h}^{2} \leq \norm{\abs{P^{\frac{1}{2}}}\abs{h}}^{2} \leq \norm{\abs{P^{\frac{1}{2}}}\eta_{\mathcal{B}}}^{2}\frac{\norm{x}^{4}}{4}\\
 &\leq \norm{\abs{P^{\frac{1}{2}}}\eta_{\mathcal{B}}}^{2}\frac{\norm{x}^{2}\nu(x)}{4\underline{\lambda}(P)}.
 \end{align*}
 We consequently have, for $(x,w)\in \mathcal{B}\times W$,
\begin{align*}
\nu(f)-\nu(x)\leq&-\epsilon \norm{x}^{2}\\
&+\sqrt{\nu(x)}\norm{\abs{P^{\frac{1}{2}}}\eta_{\mathcal{B}}}\norm{P^{\frac{1}{2}}(D_{x}f)P^{-\frac{1}{2}}}\norm{x}^{2}\\
&+
\norm{\abs{P^{\frac{1}{2}}}\eta_{\mathcal{B}}}^{2}\frac{\norm{x}^{2}\nu(x)}{4\underline{\lambda}(P)}
-d \norm{x}^{2}.
\end{align*}
For $x\in \mathbb{E}_{c_{1}}$,  $\nu(x)\leq c_{1}\leq k_{2}\Rightarrow
\nu(x)\leq { R_{\mathcal{B},i}^{2}}/{P^{-1}_{i,i}}~\forall i\in \intcc{1;n}, 
$
indicating  $x\in \mathcal{B}$ (i.e.,  $\mathbb{E}_{c_{1}}\subseteq \mathcal{B}$) and the bounds above hold over $\mathbb{E}_{c_{1}}$.  Let $(x,w)\in \mathbb{E}_{c_{1}}\times W$ (note that $\nu(x)\leq k_{1}$),  it then follows, using the definitions of $e_{1}$, and $e_{2}$, and $k_{1}$, that:
\begin{align*}
\sqrt{\nu(x)}\norm{\abs{P^{\frac{1}{2}}}\eta_{\mathcal{B}}}\norm{P^{\frac{1}{2}}(D_{x}f)P^{-\frac{1}{2}}}\norm{x}^{2}\\
+
\norm{\abs{P^{\frac{1}{2}}}\eta_{\mathcal{B}}}^{2}\frac{\norm{x}^{2}\nu(x)}{4\underline{\lambda}(P)}-d \norm{x}^{2}\\
\leq
e_{2}\sqrt{\nu(x)}\norm{x}^{2}+
e_{1} \norm{x}^{2}\nu(x)
-d \norm{x}^{2}\leq 0.
\end{align*}
Therefore,
$
\nu(f(x,w))-\nu(x)\leq -\epsilon\norm{x}^{2}\leq -({\epsilon}/{\overline{\lambda}(P)}) \nu(x),
$
and hence,
$
\nu(f(x,w))\leq \left(1-{\epsilon}/{\overline{\lambda}(P)}\right)\nu(x)$. Using induction, it can then be shown that for any $x\in \mathbb{E}_{c_{1}}$ and any disturbance sequence $\pi\in W^{\mathbb{Z}_{+}}$,
\begin{equation}\label{eq:LyapunovFunctionDecreasing}
\nu(\varphi_{x}^{\pi}(j))\leq \left(1-\frac{\epsilon}{\overline{\lambda}(P)}\right)^{j}\nu(x)~\forall j\in \mathbb{Z}_{+},
\end{equation}
implying $ \varphi_{x}^{\pi}(j)\in \mathbb{E}_{c_{1}}~~\forall j\in \mathbb{Z}_{+}$ and
\begin{align*}
\lim_{j\rightarrow \infty}\norm{\varphi_{x}^{\pi}(j)}^{2}&\leq\lim_{j\rightarrow \infty} \frac{1}{\underline{\lambda}(P)}\nu(\varphi_{x}^{\pi}(j))\\
&\leq \lim_{j\rightarrow \infty} \frac{1}{\underline{\lambda}(P)}\left(1-\frac{\epsilon}{\overline{\lambda}(P)}\right)^{j}\nu(x)=0.
\end{align*}
\end{proof}
\section{Lyapunov analysis of system \eqref{eq:cos_poly_sys}}
\label{sec:LyapunovAnalysisPolynomial}
The analysis was initially drafted with the assistance of ChatGPT \cite{chatgpt2026} and subsequently verified and refined by the authors.  
Let $\nu(x)=\|x\|^{2}$ for $x\in\mathbb{R}^{2}$ (herein, $\norm{\cdot}$ is the Euclidean norm).  
From \eqref{eq:cos_poly_sys},
$\nu(f(x,w))
=
\bigl(\cos(\nu(x))-w\nu(x)\bigr)^{2}\nu(x)$, $\in\mathbb{R}^{2}\times W$. To preserve order after squaring, we enforce positivity of 
$\cos(\nu(x))-w\nu(x)$ near the origin.  
Since $\cos(s)\ge 1-{s^{2}}/{2}$ for all $s\ge 0$,
$
\cos(\nu(x))-w\nu(x)
\ge
1-{\nu(x)^{2}}/{2}-w\nu(x)$. Thus, $\cos(\nu(x))-w\nu(x)\ge0$ whenever
$
1-{\nu(x)^{2}}/{2}-w\nu(x)\ge0,
$
i.e.,
$
0\le\nu(x)\le -w+\sqrt{w^{2}+2}$.
Because $w\mapsto -w+\sqrt{w^{2}+2}$ is decreasing on $\mathbb{R}_{+}$,
for all $w\in\intcc{w_{\min},w_{\max}}$ we have
\begin{equation}\label{eq:rho_pos}
\begin{split}
0\le\nu(x)\le\rho_{\mathrm{pos}}
:=
-w_{\max}+\sqrt{w_{\max}^{2}+2}\\
\;\Longrightarrow\;
\cos(\nu(x))-w\nu(x)\ge0 .
\end{split}
\end{equation}
Using $\cos(\tau)\le1$ for $\tau\ge0$, $
\cos(\nu(x))-w\nu(x)\le 1-w\nu(x)$.

On \eqref{eq:rho_pos}, both sides are nonnegative; hence
\[
\bigl(\cos(\nu(x))-w\nu(x)\bigr)^{2}
\le (1-w\nu(x))^{2}
\le (1-w_{\min}\nu(x))^{2}.
\]
Therefore, whenever $\nu(x)\le\rho_{\mathrm{pos}}$,
\begin{equation}\label{eq:Vk_bound_tight}
\nu(f(x,w))
\le
\nu(x)-2w_{\min}\nu(x)^2
+w_{\min}^{2}\nu(x)^3 .
\end{equation}
If additionally $\nu(x)\leq 1/w_{\min}$, then
$w_{\min}^{2}\nu(x)^3 \le w_{\min}\nu(x)^2$, yielding
\begin{equation}\label{eq:poly_dec_tight}
\nu(f(x,w))
\le
\nu(x)-w_{\min}\nu(x)^2 .
\end{equation}

Define $\rho\defas
\min\!\left\{
\rho_{\mathrm{pos}},
{1}/{w_{\min}}-\varepsilon
\right\}$,  $\varepsilon>0$, so that $w_{\min}\nu(x)<1$ for all $\nu(x)\le\rho$.
Then, \eqref{eq:poly_dec_tight} implies
$\nu(f(x,w))\le\nu(x)$ for all $\nu(x)\le\rho$,
and the sublevel set $\Omega_\rho=\{x\in\mathbb{R}^{2}:\nu(x)\le\rho\}$ is a robust invariant set. For $\nu(x)\in \intoc{0,\rho}$, and using  \eqref{eq:poly_dec_tight},
\[
\frac{1}{\nu(f(x,w))}
\ge
\frac{1}{\nu(x)-w_{\min}\nu(x)^2}
=
\frac{1}{\nu(x)}\frac{1}{1-w_{\min}\nu(x)} .
\]
Since $0\le w_{\min}\nu(x)<1$, the fact that
$
{1}/(1-z)\ge 1+z$ for $z\in \intco{0,1}$
gives
$
{1}/{\nu(f(x,w))}
\ge
{1}/{\nu(x)}+w_{\min}.
$
Along trajectories $x_{k+1}=f(x_k,w_k)$,
forward invariance ensures $\nu(x_k)\le\rho$ for all $k$,
hence
${1}/{\nu(x_{k+1})}
\ge
{1}/{\nu(x_k)}+w_{\min}$, which by induction yields
$
{1}/{\nu(x_k)}
\ge
{1}/{\nu(x_0)}+w_{\min}k$, $k\in \mathbb{Z}_{+}$. Therefore, $
\nu(x_k)
\le
\nu(x_0)/(1+w_{\min}k\,\nu(x_0))$ .
Since $\nu(x)=\|x\|^2$ ,
$
\|x_k\|
\le
{\|x_0\|}/{\sqrt{1+w_{\min}k\|x_0\|^2}},~k\in \mathbb{Z}_{+}$, proving uniform local polynomial stability of the origin for
\eqref{eq:cos_poly_sys}.

\end{document}